\documentclass[11pt]{article}

\usepackage[margin=1in]{geometry}
\usepackage{amsmath,amstext,amssymb,amsfonts}
\usepackage{amsthm}
\usepackage{times}
\usepackage{subfig}
\usepackage{graphicx}
\usepackage{psfrag}
\usepackage{color}

\newcounter{theoremc}
\newtheorem{fact}[theoremc]{Fact}

\newtheorem{clm}[theoremc]{Claim}
\newtheorem{theorem}[theoremc]{Theorem}
\newtheorem{lemma}[theoremc]{Lemma}
\newtheorem{corollary}[theoremc]{Corollary}
\newtheorem{definition}[theoremc]{Definition}

\newcommand{\mP}{\mathcal{P}}

\newcommand{\mone}{\mathbf{1}}

\newenvironment{packed_itemize}{
\begin{itemize}
  \setlength{\itemsep}{1pt}
  \setlength{\parskip}{0pt}
  \setlength{\parsep}{0pt}
}{\end{itemize}}

\newenvironment{packed_enumerate}{
\begin{enumerate}
  \setlength{\itemsep}{1pt}
  \setlength{\parskip}{0pt}
  \setlength{\parsep}{0pt}
}{\end{enumerate}}

\begin{document}

\title{The Network Improvement Problem for Equilibrium Routing\thanks{Supported in part by NSF Awards 1038578 and 1319745, an NSF CAREER Award (1254169), the Charles Lee Powell Foundation, and a Microsoft Research Faculty Fellowship.}}
\author{Umang Bhaskar\footnotemark[2]
\and
Katrina Ligett\footnotemark[2]
\and
Leonard J.\ Schulman\footnotemark[2]}
\maketitle
\renewcommand{\thefootnote}{\fnsymbol{footnote}}
\footnotetext[2]{\{umang,katrina,schulman\}@caltech.edu}
\renewcommand{\thefootnote}{\arabic{footnote}}

\begin{abstract}
In routing games, agents pick their routes through a network to minimize their own delay. A primary concern for the network designer in routing games is the average agent delay at equilibrium. A number of methods to control this average delay have received substantial attention, including network tolls, Stackelberg routing, and edge removal.

A related approach with arguably greater practical relevance is that of making investments in improvements to the edges of the network, so that, for a given investment budget, the average delay at equilibrium in the improved network is minimized. This problem has received considerable attention in the literature on transportation research and a number of different algorithms have been studied. To our knowledge, none of this work gives guarantees on the output quality of any polynomial-time algorithm. We study a model for this problem introduced in transportation research literature, and present both hardness results and algorithms that obtain nearly optimal performance guarantees.

\begin{itemize}
\item We first show that a simple algorithm obtains good approximation guarantees for the problem. Despite its simplicity, we show that for affine delays the approximation ratio of $4/3$ obtained by the algorithm cannot be improved.
\item To obtain better results, we then consider restricted topologies. For graphs consisting of parallel paths with affine delay functions we give an optimal algorithm. However, for graphs that consist of a series of parallel links, we show the problem is weakly NP-hard.
\item Finally, we consider the problem in series-parallel graphs, and give an FPTAS for this case.
\end{itemize}

Our work thus formalizes the intuition held by transportation researchers that the network improvement problem is hard, and presents topology-dependent algorithms that have provably tight approximation guarantees. 
\end{abstract}

\maketitle

\section{Introduction}
\label{sec:intro}

Routing games are widely used to model and analyze networks where traffic is routed by multiple users, who typically pick their route to minimize their delay~\cite{Roughgardenbook}. Routing games capture the uncoordinated nature of traffic routing. A prominent concern in the study of these games is the overall social cost, which is usually taken to be the average delay suffered by the players at equilibrium.
It is well known that equilibria are generally suboptimal in terms of social cost. The ratio of the average delay of the worst equilibrium routing to the optimal routing that minimizes the average delay is called the price of anarchy; tight bounds on the price of anarchy are well-studied and are known for various classes of delay functions~\cite[Chapter 18]{agtbook}.

However, the notion of price of anarchy assumes a fixed network. In reality, of course, networks change, and such changes may intentionally be implemented by the network designer to improve quality of service. This raises the question of how to identify cost-effective network improvements. Our work addresses this fundamental design problem. Specifically, given a budget for improving the network, how should the designer allocate the budget among edges of the network to minimize the average delay at equilibrium in the resulting, improved network? This crucial question arises frequently in network planning and expansion, and yet seems to have received no attention in the algorithmic game theory literature. This is surprising considering the attention given to other methods of improving equilibria, e.g., edge tolls and Stackelberg routing.

Our model of network improvement is adopted from a widely studied problem in transportation research \cite{YangH98} called the Continuous Network Design Problem (CNDP)~\cite{AbdulaalL79}. In this model, each edge in the network has a delay function that gives the delay on the edge as a function of the traffic carried by the edge. Specifically, the delay function on each edge consists of a free-flow term (a constant), plus a congestion term that is the ratio of the traffic on the edge to the conductance of the edge, raised to a fixed power. The cost to the network designer of increasing the conductance of an edge by one unit is an edge-specific constant.
Our objective is to select an allocation of the improvement budget to the edges that minimizes the social cost of equilibria in the improved network.

The continuous network design problem, along with the discrete network design problem that deals with the creation (rather than improvement) of edges, has been referred to as ``one of the most difficult and challenging problems facing transport''~\cite{YangH98}. The CNDP is generally formulated as a mathematical program with the budget allocated to each edge and the traffic at equilibrium as variables. Since the traffic is constrained to be at equilibrium, such a formulation is also called a Mathematical Program with Equilibrium Constraints (MPEC). Further, since the traffic at equilibrium is itself obtained as a solution to a optimization problem, this is also a bilevel optimization problems. Both bilevel optimization problems and MPECs have a number of other applications and have been studied independent of the CNDP as well (e.g.,~\cite{ColsonMS07}).

Owing both to the rich structure of the problem and its practical relevance, the CNDP has received considerable attention in transportation research. Because of the nonconvexity and the complex nature of the constraints, the bulk of the literature focuses on heuristics, and proposed algorithms are evaluated by performance on test data rather than formal analysis. Many of these algorithms are surveyed in~\cite{YangH98}. More recent papers give algorithms that obtain global optima~\cite{LiYZM12,LuathepSLLL11,WangL10}, but make no guarantees on the quality of solutions that can be obtained in polynomial time.

In this paper, we consider a model with fixed demands, separable polynomial delay functions on the edges and constant improvement costs. This particular model, and further restrictions of it, have been the focus of considerable attention, e.g.,~\cite{FrieszCM92,HarkerF84,Marcotte86}, and is frequently used for test instances. The model captures many of the essential characteristics of the more general problems, such as the bilevel and nonconvex nature of the problem and the equilibrium constraints. Our work thus gives the first algorithmic results with proven output quality and runtime for the network improvement problem.

\paragraph{Our Contributions.} We first focus on general graphs, and show that a simple algorithm that relaxes equilibrium constraints on the flow gives an approximation guarantee that is tight for linear delays.

\begin{packed_itemize}
\item We show that for general networks with multiple sources and sinks and polynomial delays, a simple algorithm gives an $O(d/ \log d)$-approximation to the optimal allocation, where $d$ is the maximum degree of the polynomial delay functions. If $d=1$, this gives a $4/3$-approximation algorithm.
\item We show that the approximation ratio for linear delays is tight, even for the single-commodity case: by a reduction similar to that used by Roughgarden~\cite{Roughgarden06}, we show that it is NP-hard to obtain an approximation ratio better than $4/3$.
\end{packed_itemize}

The hardness result crucially depends on the generality of the network topology. The practical relevance of the network improvement problem then motivates us to consider restricted topologies of networks, for which we give both polynomial-time algorithms and further hardness results. We restrict ourselves to single source and sink networks for the following results.

\begin{packed_itemize}
\item In graphs consisting of parallel $s$-$t$ paths with linear delays, we show that even though the problem is non-convex, first-order conditions are sufficient for optimality. We utilize this property and the special structure of the first-order conditions to give an optimal polynomial-time algorithm. If each path consists of a single edge, we give a particularly simple optimal algorithm.
\item In contrast to our previous hardness result, we show that even in graphs with linear delays and with very simple topologies consisting of a series of parallel links, called \emph{series-dipole graphs}, obtaining the optimal allocation for a given investment budget is NP-hard.
\item Lastly, in series-parallel networks with polynomial delays, we show that there exists a fully-polynomial time approximation scheme. Our algorithm is based on discretizing simultaneously over the space of flows and allocations and showing there is a near-optimal flow and allocation in the discretized space which can be obtained efficiently for series-parallel networks. The discretized flow may not correspond to the equilibrium flow; however, in series-parallel graphs, we show that the delay at equilibrium is at most the delay of the discretized flow.
\end{packed_itemize}

Our work thus presents a fairly comprehensive set of approximation guarantees for the problem of network improvement. We give tight bounds on the approximability of the problem in general networks, and optimal and near-optimal algorithms for restricted topologies. Further, we show that the problem is NP-hard --- though weakly so --- even in very restricted topologies. Our results thus supplement the work in transportation research on the problem, by formalizing the intuition that the problem is hard, and giving tight approximation algorithms for a number of cases.

\section{Related Work}
\label{sec:related}

Routing games as a model of traffic on roads were introduced by Wardrop in 1952~\cite{Wardrop52}. Beckmann et al.~\cite{BeckmannMW56} showed that equilibria in routing games are obtained as the solution to a strictly convex optimization problem if all delay functions are increasing, thus establishing the existence and uniqueness of equilibria. Wardrop's model, and our work, focuses on nonatomic routing games where the traffic controlled by each player is infinitesimal. This is a good assumption for road traffic since an individual driver has negligible impact on the delay. However, many other models of traffic in routing games are studied as well. For example, in atomic games, players can control significant traffic; this traffic may be splittable (e.g.,~\cite{HayrapetyanTW06,CatoniP91}) or unsplittable~(e.g.,~\cite{Rosenthal73}), depending on whether a player may split her traffic among multiple routes.

The problem of obtaining formal bounds on the efficiency of equilibria in games was first studied by Koutsoupias and Papadimitriou~\cite{KoutsoupiasP99}. The price of anarchy --- the ratio of the social cost at the worst equilibrium routing to the social cost of an optimal routing --- was later introduced by Papadimitriou~\cite{Papadimitriou01} as a formal measure of inefficiency. For nonatomic routing games with the social cost given by the average delay, the price of anarchy is known to be 4/3 for linear delays~\cite{RoughgardenT02}, and $\Theta(p/\log p)$ for delay functions that are polynomials of degree $p$~\cite{Roughgarden03}.

Significant research has gone into the use of tolls to improve the efficiency of routing games. It is known that tolls corresponding to the marginal delay of an optimal flow induce the optimal flow as an equilibrium~\cite{BeckmannMW56}. More generally, tolls to induce any minimal routing can be obtained as the solution to a linear program~\cite{FleischerJM04,KarakostasK04,YangH04}. A similar result for atomic splittable routing games was shown independently by Swamy~\cite{Swamy07} and Yang and Zhang~\cite{YangZ08}. These results compute efficiency without adding tolls to the delay of the players, and hence disregard the effect of tolls on the utility of the players in this computation. If tolls are added to the delays in the computation of efficiency, then even for linear delays, it is NP-hard to obtain tolls that give better than a 4/3-approximation~\cite{ColeDR06}. Since a 4/3-approximation can be obtained by not applying any tolls, this result says that it is NP-hard to find improving tolls. 

Motivated by Braess's paradox, where removal of an edge improves the efficiency of the equilibrium routing, Roughgarden~\cite{Roughgarden06} studies the problem of removing edges from a network to minimize the delay at equilibrium in the resulting network. The problem is strongly NP-hard, and there is no algorithm with an approximation ratio better than $n/2$ for general delay functions. Another method studied for improving the efficiency of routing is Stackelberg routing, which assumes that a fraction of the traffic is centrally controlled and is routed to improve efficiency. Obtaining the optimal Stackelberg routing is NP-hard even in parallel links~\cite{Roughgarden04}, although a fully-polynomial time approximation scheme is known for this case~\cite{KumarM02}.

The importance of the network improvement problem has caused it to receive significant attention in transportation research, where the version we are considering is known as the continuous network design problem'. Early research focused on heuristics that did not give any guarantees about the quality of the solution obtained. These were based on sensitivity analysis of the variational inequality to implement gradient-descent~\cite{FrieszTCM90}, as well as derivative-free algorithms~\cite{AbdulaalL79}. For a survey of other algorithms and early results on the continuous network design problem, see~\cite{YangH98}.

More recent work in the transportation literature has also tried to obtain algorithms that obtain global minima for the continuous network design problem. Early approaches include the use of simulated annealing~\cite{FrieszCM92} and genetic algorithms~\cite{Yin00}. Li et al. ~\cite{LiYZM12} reduce the problem to a sequence of mathematical programs with concave objectives and convex constraints, and show that the accumulation point of the sequence of solutions is a global optimum. If the sequence is terminated early, they show weak bounds on the quality of the solution that are consequential only under strong assumptions on the delay function and agents' demands. Wang and Lo~\cite{WangL10} reformulate the problem as a mixed integer linear program (MILP) by replacing the equilibrium constraints by constraints containing binary variables for each path, and using a number of linear segments to approximate the delay functions. This approach was further developed by Luathep et al.~\cite{LuathepSLLL11} who replaced the possibly exponentially many path variables by edge variables and gave a cutting constraint algorithm for the resulting MILP. The last two methods were shown to converge to a global optimum of the linearized approximation in finite time, but require solving a MILP with a possibly exponential number of variables and constraints.

A variant of the problem where the initial conductance of every edge in the network is zero, and the budget is part of the objective rather than a hard constraint, is studied by Marcotte~\cite{Marcotte86} and, independent of our work, by Gairing et al.~\cite{GairingHK13}. Unlike the work cited earlier, these papers give provable guarantees on the performance of polynomial-time algorithms. Marcotte gives an algorithm that is a 2-approximation for monomial delay functions and a $5/4$-approximation for linear delay functions. Gairing et al. present an algorithm that improves upon these upper bounds, give an optimal polynomial-time algorithm for single-commodity instances, and show that the problem is APX-hard in general. In our problem, the budget is a hard constraint, and edges may have arbitrary initial capacities. Our problem is demonstrably harder than this variant: e.g., in contrast to the polynomial-time algorithm for single-commodity instances given by Gairing et al.~\cite{GairingHK13}, we show that in our problem no approximation better than $4/3$ is possible even in single-commodity instances.

\section{Notation and Preliminaries}
\label{sec:notation}

$G=(V,E)$ is a directed graph with $|E|=m$ and $|V|=n$. If $G$ is a two-terminal graph, then it has two special vertices $s$ and $t$ called the source and the sink, collectively called the terminals. A $u$-$v$ path $p=((v_0,v_1),(v_1, v_2),\dots,(v_{k-1},v_k))$ is a sequence of edges with $v_0 = u$, $v_k = v$ and edges $(v_i,v_{i+1}) \in E$. In a two-terminal graph, each edge lies on an $s$-$t$ path, and we use $\mP$ to denote the set of all $s$-$t$ paths. Given vertices $s'$, $t'$ in graph $G$, vector $(f_e)_{e \in E}$ is an $s'$-$t'$ flow of value $d$ if the following conditions are satisfied:
\begin{align*}
\sum_{(u,w) \in E} f_{uw} - \sum_{(w,u) \in E} f_{wu} & = 0, ~ \forall u \in V \setminus \{s',t'\} \\
\sum_{(s,w) \in E} f_{sw} - \sum_{(w,s) \in E} f_{ws} & = d \\
f_e & \ge 0, ~ \forall e \in E \, .
\end{align*}

We use $|f|$ to denote the value of flow $f$. A path decomposition of an $s'$-$t'$ flow $f$ is a set of flows $\{f_p\}$ along $s'$-$t'$ paths $p$ that satisfies $f_e = \sum_{p: e \in p} f_p$, $\forall e$. A path decomposition for flow $f$ so that $f_p > 0$ for at most $m$ paths can be obtained in polynomial time~\cite{AhujaMO93}. Without reference to a path decomposition, we use $f_p > 0$ to indicate that $f_e > 0$ for all $e \in p$.

Each edge $e \in E$ has an increasing delay function $l_e(x)$ that gives the delay on the edge as a function of the flow on the edge. For flow $f$ and path $p$, $l_p(f) := \sum_{e \in p} l_e(f_e)$ is the delay on path $p$. Further, $f_e l_e(f_e)$ is the total delay on edge $e$, and the total delay of flow $f$ is $\sum_{e \in E} f_e l_e(f_e)$.

\paragraph{Routing games.} A routing game is a tuple $\Gamma = (G,l,K)$ where $G$ is a directed graph, $l$ is a vector of delay functions on edges, and $K = \{s_i, t_i, d_i\}_{i \in I}$ is a set of triples where $d_i$ is the total traffic routed by players of commodity $i$ from $s_i$ to $t_i$. Each player of commodity $i$ in a routing game controls infinitesimal traffic and picks an $s_i$-$t_i$ path $p$ on which to route her flow, as her strategy. The strategies induce a flow $f^i$. Let $f = \sum_i f^i$, then the delay of a player that selects path $p$ as her strategy is $l_p(f)$. In the single-commodity case, $|I| = 1$. We say a flow $f$ is a valid flow for routing game $\Gamma$ if $f = \sum_{i \in I} f^i$ where each $f^i$ is an $s_i$-$t_i$ flow of value $d_i$.

At equilibrium in a routing game, each player minimizes her delay, subject to the strategies of the other players. The equilibrium in a routing game is also called a Wardrop equilibrium.

\begin{definition}
A set of flows $\{f^i\}_{i \in I}$ where $f^i$ is an $s_i$-$t_i$ flow of value $d_i$ is a Wardrop equilibrium if for all $i \in I$, for any $s_i$-$t_i$ paths $p$, $q$ such that $f_p^i > 0$,  $l_p(f) \le l_q(f)$.\label{defn:wardrop}
\end{definition}

We use \emph{equilibrium flow} to refer to the set of flows $\{f^i\}_{i \in I}$ that form a Wardrop equilibrium. The equilibrium flow is also obtained as the solution to the following mathematical program. Since the delay functions are increasing, the program has a strictly convex objective with linear constraints, and hence the first-order conditions are necessary and sufficient for optimality. Further, because of strict convexity, the equilibrium flow is unique.

\[
\min \sum_{e \in E} \int_0^{f_e} l_e(x) \, dx, ~ \mbox{s.t. $f = \sum_{i \in I} f^i$ and $f^i$ is an $s_i$-$t_i$ flow of value $d_i$} \, .
\]

Definition~\ref{defn:wardrop} then corresponds to the first-order conditions for optimality of the convex program. By Definition~\ref{defn:wardrop}, each $s^i$-$t^i$ path $p$ with $f_p^i > 0$ has the same delay at equilibrium. Let $L^i$ be this common path delay. Then the total delay $\sum_e f_e l_e(f_e) = \sum_i d_i \, L^i$, where $f = \sum_{i \in I} f^i$ and $\{f^i\}_{i \in I}$ is the equilibrium flow. The average delay is $\sum_i d_i \, L^i / \sum_i d_i$. 

\paragraph{Network Improvement.} In the network improvement problem, we are given a routing game $\Gamma$, where the delay function on each edge $e$ is of the form $l_e(x) = (x/c_e)^{n_e} + b_e$. We call $c_e$ the \emph{conductance}, $1/c_e$ the \emph{resistance}, and $b_e$ the \emph{length} of edge $e$. We assume $c_e \ge 0$ and $n_e > 0$, and hence the delay is an increasing function of the flow on the edge. The delay function on an edge is affine if $n_e = 1$.  Each edge has a marginal cost of improvement, $\mu_e$.  Upon spending $\beta_e$ to improve edge $e$, the conductance of the edge increases to $c_e + \mu_e \beta_e$. For a given budget $B$, a valid allocation is a vector $\beta = (\beta_e)_{e \in E}$ so that $\sum_e \beta_e \le B$ and $\beta_e \ge 0$ for each $e \in E$. The objective is to determine a valid allocation of the budget $B$ to the edges to minimize the average delay obtained at equilibrium with the modified delay functions $l_e(x, \beta_e) = \left(x/(c_e + \mu_e \beta_e)\right)^{n_e} + b_e$. Delay functions are affine if $n_e = 1$ on all edges.

Let $\beta = (\beta_e)_{e \in E}$ be the vector of edge allocations. Since the flow at equilibrium is unique, for any $\beta$, the average delay at equilibrium is unique. $L(\beta)$ is this unique average delay as a function of the edge allocations. When considering a flow $f$ other than the equilibrium flow, we use $L(f, \beta)$ to denote the average delay of flow $f$ with the modified delay functions. We will also have occasion to allocate budget to units other than edges, e.g., paths, and will slightly abuse notation to express the average delay in terms of these units.

Our problem corresponds to the following (non-linear, possibly non-convex) optimization problem:

\begin{equation}
\min_\beta L(\beta), ~ \mbox{s.t. } \sum_e \beta_e \le B, ~ \beta_e \ge 0 ~ \forall e \in E \, . \label{eqn:main}
\end{equation}

We use $\beta^*$ to denote an optimal solution for this problem, and define $L^* := L(\beta^*)$. As is common in nonlinear optimization, instead of an exact solution we will obtain a solution that is within a specified additive tolerance of $\epsilon$ of the exact solution, i.e., a valid allocation $\hat{\beta}$ so that $L(\hat{\beta}) - \epsilon \le L(\beta^*)$. An algorithm is polynomial-time if it obtains such a solution in time polynomial in the input size and $\log (1/\epsilon)$. Since the problem has linear constraints, the first-order conditions are necessary for optimality~(e.g.,~\cite{Ye05}). By the first-order conditions for optimality, for any edges $e$ and $e'$,

\begin{align}
\beta_e > 0 & \Rightarrow \frac{\partial L(\beta)}{\partial \beta_e} \le \frac{\partial L(\beta)}{\partial \beta_{e'}} \, . \label{eqn:foc}
\end{align}

For any edge $e$ and allocation $\beta$, define $c_e(\beta) = c_e + \mu_e \beta_e$. For a path $p$, $b_p = \sum_{e \in p} b_e$ as the length of path $p$. For affine delay functions, define $c_p(\beta) = 1/ \sum_{e \in p} \frac{1}{c_e(\beta)}$ as the conductance of path $p$, and the resistance of path $p$ as the reciprocal of the conductance: $r_p(\beta) = 1/c_p(\beta)$. For $k \in \mathbb{Z}_+$, $[k] := \{1, 2, \dots, k\}$.  The following statement is easily verified; it is used often in our proofs, hence we state it formally below for ease of reference.

\begin{fact}
For $x$, $y$, $z$ $\in \mathbb{R}_{\ge 0}$ and $k \in \mathbb{R}_{> 0}$, 
\[
 \frac{x}{y} > k \iff \frac{x+kz}{y+z} < \frac{x}{y} \iff \frac{x-kz}{y-z} > \frac{x}{y} \, .
\]
\label{fact:xyz}
\end{fact}

\section{Tight Bounds in General Graphs}
\label{sec:general}

\subsection{Simple approximation algorithms for general graphs}

We start with a simple polynomial-time algorithm that gives a good approximation for the general network improvement problem: with multiple sources and sinks, in general graphs, and with polynomial delay functions. The algorithm follows from the observation that relaxing the equilibrium constraints on the flow yields a convex problem. Although the algorithm described here is a very natural one and lower bounds for its performance were given in~\cite{Marcotte86}, the upper bounds shown here appear not to have been noticed earlier. In fact, in the next section we will show that the bound obtained by this simple algorithm is tight in the case of affine delay functions.

Consider the following problem:

\[
\mbox{COPT: } ~ \min_{x, \beta} \sum_e x_e l_e(x_e, \beta_e) \mbox{ s.t. $x$ is a valid flow, and $\beta$ is a valid allocation.}
\]

\noindent It is obvious the constraints for COPT are convex. We now show that if the delay functions $l_e(x_e, \beta_e)$ are polynomial, then the objective is a convex function as well.

\begin{lemma}
The objective function for COPT with polynomial delays is convex.
\label{lem:coptconvex}
\end{lemma}

\begin{proof}
We will show that the Hessian for each term in the summation in the objective is positive semi-definite. This will show that each individual term is convex, and hence the objective is convex as well. Each term in the summation is of the form

\[
l(x, \beta) := \frac{x^{n+1}}{(c + \mu \beta)^n} + bx \, .
\]

\noindent We will show that $l(x, \beta)$ is a convex function for the proof. The following derivatives are easily obtained:

\begin{align}
\frac{\partial l}{\partial x} = \frac{(n+1) x^n}{(c+\mu \beta)^n} + b \, , ~ \frac{\partial^2 l}{\partial x^2} = \frac{(n+1)nx^{n-1}}{(c+\mu \beta)^n} \label{eqn:pdx} \\
\frac{\partial l}{\partial \beta} = - \frac{n \mu x^{n+1}}{(c+\mu \beta)^{n+1}} \, , ~ 
\frac{\partial^2 l}{\partial \beta^2} = \frac{(n+1)n\mu^2 x^{n+1}}{(c+\mu \beta)^{n+2}} \label{eqn:pdbeta} \\
\frac{\partial^2 l}{\partial \beta \, \partial x} = \frac{\partial^2 l}{\partial x \, \partial \beta} = - \frac{n(n+1) \mu x^n}{(c+ \mu \beta)^{n+1}} \label{eqn:pdxbeta}
\end{align}

The Hessian for function $l(x, \beta)$ is given by

\[
H := \left[ \begin{array}{ll}
		\frac{\partial^2 l}{\partial x^2} & \frac{\partial^2 l}{\partial x \partial \beta} \\
		\frac{\partial^2 l}{\partial \beta \partial x} & \frac{\partial^2 l}{\partial \beta^2} 
	\end{array} \right]
\]

For any vector $\alpha = [ \alpha_1 ~ \alpha_2 ]^T$ we will show that $\alpha^T H \alpha \ge 0$, which proves the positive semi-definiteness of $H$ and hence the convexity of $l(x, \beta)$. By the expressions from~(\ref{eqn:pdx}),~(\ref{eqn:pdbeta}) and~(\ref{eqn:pdxbeta}),

\begin{align*}
\alpha^T H \alpha & = \alpha_1^2 \frac{n (n+1) x^{n-1}}{(c+\mu \beta)^n} + \alpha_2^2 \frac{n(n+1)x^{n+1} \mu^2}{(c+\mu \beta)^{n+2}} - 2 \alpha_1 \alpha_2  \frac{n(n+1)\mu x^n}{(c+\mu \beta)^{n+1}} \\
	& = \frac{n(n+1) x^{n-1}}{(c+\mu \beta)^n} \left( \alpha_1^2 + \alpha_2^2 \frac{x^2 \mu^2}{(c+\mu \beta)^2} - 2 \alpha_1 \alpha_2 \frac{x\mu}{(c+ \mu \beta)} \right) \\
	& = \frac{n(n+1) x^{n-1}}{(c+\mu \beta)^n} \left( \alpha_1 - \alpha_2 \frac{x \mu}{c+\mu \beta} \right)^2 \ge 0 \, .
\end{align*}

\end{proof}

The solution $(\hat{x}, \hat{\beta})$ to problem COPT can thus be obtained in polynomial time. Our approximation algorithm then simply returns the allocation $\hat{\beta}$ obtained by solving COPT.

\begin{lemma}
If the delay function on every edge is affine, then the delay obtained at equilibrium for the allocation $\hat{\beta}$ is a $4/3$-approximation to the  delay at equilibrium for the optimal allocation $\beta^*$. In general if all delay functions are polynomials of degree at most $p$, then the delay obtained at equilibrium for $\hat{\beta}$ is an $O(p/\log p)$-approximation to the delay at equilibrium for the optimal allocation $\beta^*$.
\end{lemma}

\begin{proof}
Let $f^*$ be the equilibrium flow obtained for the optimal allocation $\beta^*$, and let $f$ be the equilibrium flow for allocation $\hat{\beta}$. It is obvious that $\sum_e f_e^* l_e(f_e^*, \beta_e^*) \ge \sum_e \hat{x}_e l_e (\hat{x}_e, \hat{\beta}_e)$. For fixed affine delays, it is well-known that the total delay of the equilibrium routing is at most $4/3$ that of the the flow that minimizes the total delay~\cite{RoughgardenT02}. Thus, $\sum_e f_e l_e(f_e, \hat{\beta}_e) \le 4/3 \sum_e \hat{x}_e l_e (\hat{x}_e, \hat{\beta}_e)$. The statement of the lemma for affine delays follows. Further, the statement for general polynomial delays follows since the total delay of the equilibrium routing is known to be at most $O(p/ \log p)$ that of the the flow that minimizes the total delay~\cite{Roughgarden03}.
\end{proof}

\subsection{A nearly tight lower bound for affine delays}

We now show that the upper bounds obtained in the previous section are tight for affine delays, even for single-commodity routing games. We give a reduction from the problem of 2-Directed Disjoint Paths, which is known to be NP-complete~\cite{FortuneHW80}:

\begin{definition}[2-Directed Disjoint Paths (2DDP)]
Given a directed graph $G$ and vertices $s_1$, $s_2$, $t_1$ and $t_2$, do there exist $s_i$-$t_i$ paths $p_i$ such that $p_1$ and $p_2$ are vertex-disjoint?
\end{definition}

Note that the 2DDP problem is known to be solvable in polynomial-time if the graph is acyclic~\cite{ShiloachP78} or planar~\cite{Schrijver94}. Our reduction is essentially identical to that given by Roughgarden~\cite{Roughgarden06} for the problem of removing edges from a network to improve the total delay at equilibrium in the resulting network.

In our reduction, we allow the budget to be unbounded. We modify the graph for the 2DDP problem by adding vertices $s$, $t$ and edges $(s,s_1)$, $(s,s_2)$, $(t_1, t)$ and $(t_2,t)$. For all edges except the ones added, we choose $c_e = 0$, $b_e = 0$ and $\mu_e = 1$. Thus none of these edges cannot be used at equilibrium unless it is given a strictly positive allocation. For edges $(s,s_1)$ and $(t_2,t)$, we choose $c_e = 1$, $b_e = 1$, and $\mu_e = 1$. For edges $(s,s_2)$ and $(t_1,t)$, we choose $c_e = 1$, $b_e = 0$ and $\mu_e = 0$; thus any allocation to these edges does not affect the delay function. The demand between $s$ and $t$ is 1.

We now show that if the given instance of 2DDP contains two disjoint paths, then there exists an allocation that yields two vertex-disjoint $s$-$t$ paths with delay functions $x+1$ on both, and thus an average delay of $3/2$. If two vertex-disjoint paths do not exist, any allocation has average delay at least $2$ because the existence of a common vertex leads to inefficient routing, similar to that in the Braess graph.

\begin{figure}[h]
\psfrag{G}{$\displaystyle G$}
\psfrag{s}{$\displaystyle s$}
\psfrag{t}{$\displaystyle t$}
\psfrag{v}{$\displaystyle v$}
\psfrag{s1}{$\displaystyle s_1$}
\psfrag{s2}{$\displaystyle s_2$}
\psfrag{t1}{$\displaystyle t_1$}
\psfrag{t2}{$\displaystyle t_2$}
\psfrag{ss1}{$\displaystyle \frac{x}{1+\beta}$}
\psfrag{t2t}{$\displaystyle\frac{x}{1+\beta}$}
\psfrag{ss2}{$\displaystyle x$}
\psfrag{t1t}{$\displaystyle x$}
\centering \includegraphics[scale=0.5]{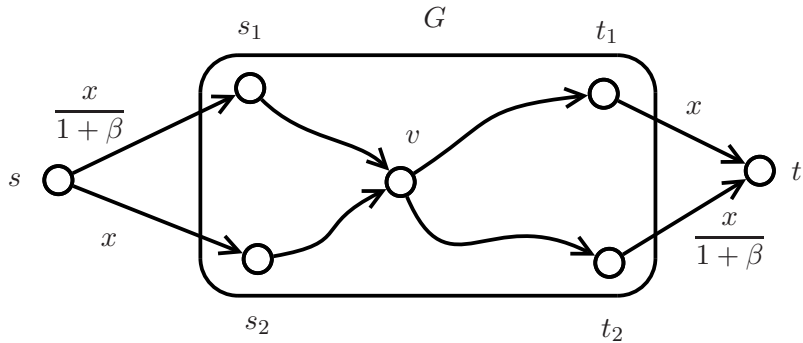}
\caption{The graph for reduction from 2DDP. Each of the edges in $G$ have $c_e = b_e = 0$ and $\mu_e = 1$. Edges $(s,s_2)$ and $(t_1,t)$ have $\mu_e = 0$.}
\end{figure}

\begin{lemma}
If $G$ contains two disjoint paths, then there is an allocation for the CNDP instance constructed with average delay $3/2$, otherwise any allocation has average delay at least $2$.
\label{lem:2ddphard}
\end{lemma}

\begin{proof}
If there exist two disjoint paths, we allocate infinite budgets to exactly the edges on these paths, thus reducing the delay functions on these edges to zero. We additionally allocate infinite budgets to edges $(s,s_1)$ and $(t_2,t)$. Then the flow that routes $1/2$ on the $s$-$s_1$-$t_1$-$t$ path and $1/2$ on the $s$-$s_2$-$t_2$-$t$ path is an equilibrium flow of average delay $3/2$.

Suppose for a contradiction that there do not exist two disjoint paths but the average delay at equilibrium for an allocation is less than $2$. Let $F$ be the subset of edges which carry strictly positive flow at equilibrium. Then $F$ must contain an $s_1$-$t_1$ path, as well as an $s_2$-$t_2$ path. To see this, note that $F$ cannot contain an  $s$-$s_1$-$t_2$-$t$ path, since the delay on this path would be at least 2. However, both $(s,s_1)$ and $(t_2,t)$ must carry positive flow. Therefore, there must be an $s_1$-$t_1$ path and an $s_2$-$t_2$ path in $F$. Since these paths cannot be vertex-disjoint, let $v$ be the common vertex. Then any $s$-$v$ path in $F$ must have delay at least 1, and any $v$-$t$ path must similarly have delay at least 1. Hence the delay at equilibrium must be at least 2.
\end{proof}
\section{Single and Parallel Paths}
\label{sec:singlepath}

\subsection{Single Paths} 
\label{sec:singlepath1}

We first consider the case where $G$ is a simple $s$-$t$ path. In this case, we show that the delay at equilibrium $L(\beta)$ is a convex function. Thus, obtaining the optimal allocation requires minimizing a convex function subject to linear constraints, which can be done in polynomial time by, e.g., interior-point methods~\cite{BoydV04}.

\begin{lemma}
If $G$ is an $s$-$t$ path, then $L(\beta)$ is convex.
\end{lemma}

\begin{proof} For an $s$-$t$ path $p$, the delay at equilibrium $L(\beta)$ for an allocation $\beta$ is $L(\beta) = \sum_{e \in p} \frac{d}{c_e(\beta)} + b_e$. Since $d$ and $b_e$ are fixed, minimizing $L(\beta)$ is equivalent to minimizing $\sum_{e \in p} 1/c_e(\beta)$, which is a convex function, since each $c_e(\beta)$ is an affine function. 
\end{proof}

\subsection{Parallel Paths}
\label{sec:ppaths}

When $G$ consists of parallel paths between $s$ and $t$, $L(\beta)$ may not be a convex function of $\beta$, and Figure~\ref{fig:nonconvex} gives an example of this nonconvexity. The graph shows the delay at equilibrium as the allocation to edge 1 is increased and allocation to edge 2 is decreased, keeping the total allocation equal to the budget $B=3$. We will show that for network improvement, the first-order conditions for optimality are sufficient. Hence, any solution that satisfies the first-order optimality conditions is a global minimum. We will then use this characterization to give a continuous greedy-like algorithm that uses the particular structure of the first-order optimality conditions to obtain an allocation.

\begin{figure}[h]
\centering
\subfloat[][The example.]{
\psfrag{l1}{$10x+90$}
\psfrag{p1}{$\mu_1 = 1$}
\psfrag{l2}{$5x$}
\psfrag{p2}{$\mu_2 = 0.1$}
\psfrag{d}{$d = 40$}
\raisebox{0.3in}{\includegraphics[scale=0.2]{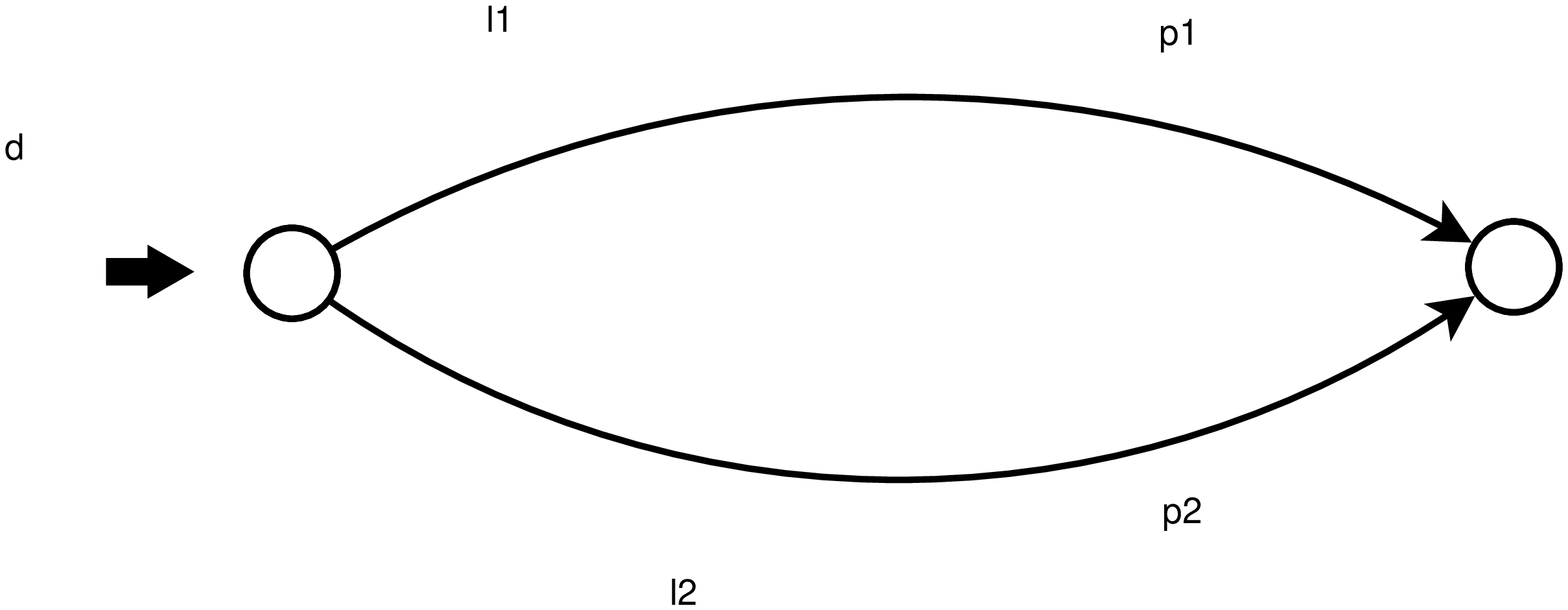}}
\label{fig:nonconvex1}}
\subfloat[][Graph showing nonconvexity of equilibrium delay.]{\includegraphics[scale=0.6]{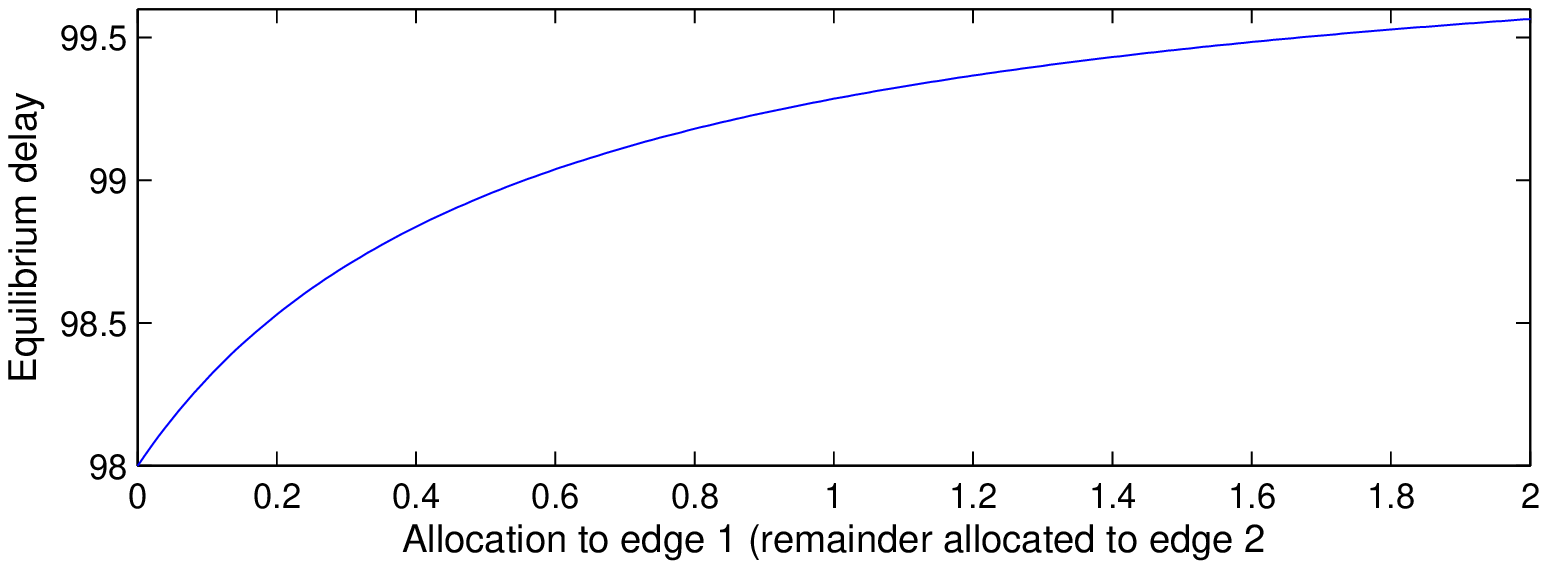}
\label{fig:nonconvex2}}
\caption{Example showing nonconvexity of equilibrium delay. Budget $B=3$.} \label{fig:nonconvex}
\end{figure}

Assume the given graph consists of $m$ $s$-$t$ paths. As before, $\mP$ is the set of all $s$-$t$ paths. We relabel paths so that $b_1 \le b_2 \le \dots \le b_m$ where $b_i$ is the length of path $i$, and for simplicity, we assume that all these inequalities are strict (the case where some inequalities are not strict requires only minor modifications to the algorithm and analysis but requires tedious notation). Here, we assume that the equilibrium flow obtained for the optimal allocation $L(\beta^*) > b_m$, and hence by definition of the equilibrium flow all paths must carry positive flow at equilibrium. Relaxing the assumption requires us to solve the problem multiple times for increasing subsets of $\mP$, and we leave this to the appendix.

By Section~\ref{sec:singlepath1}, we know how to maximize the conductance of any path $p$ for a fixed budget. Hence in the case of parallel paths, we will focus on obtaining an allocation to paths, rather than individual edges on paths, and assume that once the optimal allocation to paths is known, we can compute the allocation to the edges. We use $\beta_p$ to denote the (scalar) allocation to path $p$, $\beta$ to denote the vector of allocations to paths, and use $c_p(\beta_p)$ to denote the maximum conductance of path $p$ obtained for allocation $\beta_p$.

We first establish the concavity of path conductance.\footnote{Note that $1/c_p(\beta)$ is a convex function, since each $1/c_e(\beta)$ is convex. However, this does not imply that $c_p(\beta)$ is concave, since the reciprocal of a convex function is not necessarily concave, even in the single variable case. For example, both $x^2$ and $1/x^2$ are convex for $x > 0$.}

\begin{lemma}
For any path $p$, let $\beta'$ and $\beta''$ be two vectors of allocations to edges in path $p$ and define $\beta(\lambda) := \lambda \beta' + (1-\lambda) \beta''$ for $0 \le \lambda \le 1$ and $c_p(\lambda) := c_p(\beta(\lambda))$. Then
\begin{packed_enumerate} 
\item $c_p(\lambda)$ is concave.
\item If $c_p(\lambda)$ is not strictly concave at some $\lambda \in (0,1)$, then $c_p(\lambda)$ is linear for all $\lambda \in (0,1)$.
\end{packed_enumerate}
\label{lem:cpconcave}
\end{lemma}

\begin{proof}
We will find it more convenient in the proof to work with resistances rather than conductances, where $r_p(\lambda) = 1/c_p(\lambda)$ and for each edge $e$,  $r_e(\lambda) = 1/c_e(\lambda)$. Then

\begin{align*}
\frac{d c_p(\lambda)}{d \lambda} & = - \frac{1}{r_p^2(\lambda)} \frac{d r_p(\lambda)}{d \lambda} 
\end{align*}

\noindent and differentiating again,

\begin{align}
\frac{d^2 c_p(\lambda)}{d \lambda^2} & = - \frac{1}{r_p^4(\lambda)} \left( r_p^2(\lambda) \frac{d^2 r_p(\lambda)}{d \lambda^2} - 2 r_p(\lambda) \left( \frac{d r_p(\lambda)}{d \lambda} \right)^2 \right) \nonumber \\
 &  = - \frac{1}{r_p^3(\lambda)} \left( r_p(\lambda) \frac{d^2 r_p(\lambda)}{d \lambda^2} - 2 \left( \frac{d r_p(\lambda)}{d \lambda} \right)^2 \right)\, . \label{eqn:cprp}
\end{align}

\noindent We will use Cauchy-Schwarz and the expression for $d^2 r_p(\lambda)/d \lambda^2$ to show that the expression on the right is nonpositive, which will complete the proof of the lemma. For any edge $e$, $c_e(\lambda)$ is a linear function of $\lambda$, and hence $d^2 c_e(\lambda) / d \lambda^2 = 0$. Equation~(\ref{eqn:cprp}) holds for edges as well as paths, and hence

\begin{align}
\frac{d^2 r_e(\lambda)}{d \lambda^2} & = \frac{2}{r_e(\lambda)} \left( \frac{d r_e(\lambda)}{d \lambda} \right)^2 \, .
\label{eqn:ce}
\end{align}

To use Cauchy-Schwarz, define

\[
x_e := \frac{d^2 r_e(\lambda)}{d \lambda^2} ~ \mbox{ and } ~ y_e := r_e(\lambda) \, .
\]

Since $r_p(\lambda) = \sum_{e \in p} r_e(\lambda)$,

\begin{align}
r_p(\lambda) \frac{d^2 r_p(\lambda)}{d \lambda^2} = \left( \sum_e y_e \right) \left( \sum_e x_e \right) \ge \left( \sum_e \sqrt{x_e y_e} \right)^2 = 2 \left( \frac{d r_p(\lambda)}{d \lambda} \right)^2 \label{eqn:xeye}
\end{align}

\noindent where the inequality is by Cauchy-Schwarz and the second equality follows since $\sqrt{x_e y_e} = \sqrt{2}\, d r_e(\lambda)/d \lambda$ from~(\ref{eqn:ce}). Replacing in~(\ref{eqn:cprp}), we get the proof of the first part of the lemma.

For the second part of the proof, assume that $c_p(\lambda)$ is not strictly concave at $\bar{\lambda}$, i.e., $d^2 c_p(\bar{\lambda})/d \lambda^2 = 0$. Then the inequality in~(\ref{eqn:xeye}) must be an equality, which again by Cauchy-Schwarz is possible if and only if the vectors $x$ and $y$ are parallel, i.e., $y_e = k x_e$ for all $e \in E$ and some constant $k$. Thus $r_e(\bar{\lambda}) = k d^2 r_e(\bar{\lambda})/d \bar{\lambda}^2$, or from~(\ref{eqn:ce}),

\[
r_e(\bar{\lambda}) = k' \frac{d r_e (\bar{\lambda})}{d \bar{\lambda}}
\]

\noindent where $k' = \sqrt{2k}$. Since $r_e(\lambda) = 1/c_e(\lambda)$, this is equivalent to

\begin{align}
\frac{1}{c_e(\bar{\lambda})} = - \frac{k'}{c_e^2(\bar{\lambda})} \frac{d c_e (\bar{\lambda})}{d \lambda} \, , ~ \mbox{ or } ~ c_e(\bar{\lambda}) = - k' \frac{d c_e (\bar{\lambda})}{d \lambda} \, . \label{eqn:cpceratio}
\end{align}

\noindent We will show now that the vectors $(c_e(\lambda))_{e \in p}$ and $(\frac{d c_e (\lambda)}{d \lambda})_{e \in p}$ are parallel for all $\lambda \in [0,1]$. Hence the inequality in~(\ref{eqn:xeye}) is always an equality, and $d^2 c_p(\lambda)/d \lambda^2 = 0$ for all $\lambda$.

For any $\lambda$, since $c_e(\lambda)$ is a linear function, 

\[
c_e(\lambda) = c_e(\bar{\lambda}) + (\lambda - \bar{\lambda}) \frac{d c_e(\lambda)}{d \lambda} = (\lambda - \bar{\lambda}-k') \frac{d c_e (\lambda)}{d \lambda} 
\]

\noindent where the second equality follows from~(\ref{eqn:cpceratio}). Thus the vectors $(c_e(\lambda))_{e \in p}$ and $(\frac{d c_e (\lambda)}{d \lambda})_{e \in p}$ are parallel, and hence $d^2 c_p(\lambda)/d \lambda^2 = 0$ for all $\lambda$.
\end{proof}

The following corollary immediately follows from the first part of the lemma.

\begin{corollary}
For any path $p$ and vector $\beta$ of allocations to the edges of $p$, $c_p(\beta)$ is concave.
\label{cor:cpconcave}
\end{corollary}

We now obtain an expression for the delay at equilibrium. For an allocation $\beta$, let $x$ be the flow at equilibrium and $\{x_p\}_{p \in \mP}$ be the unique flow decomposition. Then $x_p>0$ iff $L(\beta) > b_p$. Since $L(\beta^*) > b_p$, for each path $p \in \mP$, by definition of equilibria,

\begin{align}
L(\beta) & = \frac{x_p}{c_p(\beta)} + b_p \, . \label{eqn:peflow}
\end{align}

\noindent Multiplying both sides by $c_p(\beta)$, and summing over all paths yields

\begin{align}
L(\beta) & = \frac{d + \sum_{p \in \mP} c_p(\beta) b_p}{\sum_{p \in \mP} c_p(\beta)} \, . \label{eqn:pedelay}
\end{align}

We now show that if $L(\beta^*) > b_m$, the first order conditions for optimality are also sufficient. 

\begin{lemma}
Let $\beta'$ and $\beta''$ be two valid allocations where $\beta'$ satisfies the first-order conditions for optimality, and define $\beta(\lambda) := (1-\lambda)\beta' + \lambda \beta''$ and $L(\lambda) = L(\beta(\lambda))$. Then either $L(0) \le L(\lambda)$ for all $\lambda \in [0,1]$, or there is a valid allocation $\beta$ so that $L(\beta) \le b_m$.
\label{lem:ppsufficient}
\end{lemma}

\begin{proof}
Our proof proceeds by considering all stationary points in $\lambda \in [0,1]$. We will show that if $L(\lambda) > b_m$ for all $\lambda \in [0,1]$, then either any stationary point is a minima, or $L(\lambda)$ is constant in the interval $[0,1]$. In the former case, since there are no maxima, and maxima and minima must alternate, $L(0)$ is the only minima in $[0,1]$, and hence in either case, $L(0) \le L(\lambda)$ for all $\lambda \in [0,1]$.

Let $\lambda'$ be a stationary point. Then $dL(\lambda')/d \lambda = 0$. We first show that $d^2L(\lambda')/d\lambda^2 \ge 0$, and hence there are no maxima. From~(\ref{eqn:M}), and since $L$ is a function of $\lambda$ rather than $\beta$, for any path $q \in \mP$,

\begin{align*}
\frac{\partial L(\lambda)}{\partial c_q(\lambda)} & = \frac{b_q \left( \sum_{p \in \mP} c_p(\lambda) \right) - \left(d + \sum_{p \in \mP} c_p(\lambda) b_p\right)}{\left(\sum_{p \in \mP} c_p(\lambda) \right)^2} ~ = ~ \frac{1}{\sum_{p \in \mP} c_p(\lambda)} \left(b_q - L(\lambda) \right) \, .
\end{align*}

\noindent and hence, by the chain rule,

\begin{align}
\frac{d L(\lambda)}{d \lambda} & = \sum_{p \in \mP} \frac{\partial L(\lambda)}{\partial c_p(\lambda)} \frac{d c_p(\lambda)}{d \lambda}  = \frac{1}{\sum_{p \in \mP} c_p(\lambda)} \left(\sum_{p \in \mP} \left(b_p - L(\lambda) \right) \frac{d c_p(\lambda)}{d \lambda} \right) \, .
\label{eqn:ppfirstderivative}
\end{align}

Define $A(\lambda)$ to be the term in parentheses in~(\ref{eqn:ppfirstderivative});  then $\frac{dL(\lambda)}{d \lambda} = A(\lambda)/\sum_{p \in \mP} c_p(\lambda)$. Note that $A(\lambda) = 0$ if and only if $\frac{dL(\lambda)}{d \lambda} = 0$, and hence $A(\lambda') = 0$. Further,  for the second derivative, we get

\begin{align}
\frac{d^2 L(\lambda)}{d \lambda^2} & = \frac{1}{\left( \sum_{p \in \mP} c_p(\lambda) \right)^2} \left( \sum_{p \in \mP} c_p(\lambda) \frac{d A(\lambda)}{d \lambda}  - A(\lambda) \sum_{p \in \mP} \frac{d c_p(\lambda)}{d\lambda} \right) \nonumber \\
\end{align}

\noindent and since $A(\lambda') = 0$,

\begin{align}
\frac{d^2 L(\lambda')}{d \lambda^2}	& = \frac{1}{\sum_{p \in \mP} c_p(\lambda')}\frac{d A(\lambda')}{d \lambda} \, . \label{eqn:ppsecondderivative}
\end{align}

We will now show that $A(\lambda')$ is nondecreasing, and hence any stationary point cannot be a maxima. Each term in the summation for $A(\lambda)$ is the product of $b_p - L(\lambda)$ and $d c_p(\lambda)/d\lambda$. By assumption, $\frac{dL(\lambda')}{d \lambda} = 0$, hence $b_p - L(\lambda')$ is constant and negative. By Corollary~\ref{cor:cpconcave}, the second term is nonincreasing, hence the product is nondecreasing. Each of the summands is nondecreasing, and hence $A(\lambda')$ must be nondecreasing.

Further, if $d^2L(\lambda')/d \lambda^2 = 0$, then $A(\lambda')$ must be a constant by~(\ref{eqn:ppsecondderivative}). Since each summand is nondecreasing, each summand must in fact be constant, and in particular $d c_p(\lambda')/d \lambda$ must be constant, i.e., $c_p(\lambda')$ must be linear. However, in this case, by the second part of Lemma~\ref{lem:cpconcave}, $c_p(\lambda)$ is linear for $\lambda \in (0,1)$. Hence the second derivative is zero in $(0,1)$, which by integration, and since $dL(\lambda')/d \lambda = 0$, forces the first derivative to be zero in $(0,1)$; hence, $L(\lambda)$ is constant in $[0,1]$.
\end{proof}

\paragraph*{An optimal algorithm.}

We now describe an algorithm for minimizing the delay at equilibrium on parallel paths. In order to describe our algorithm to optimize $L(\beta)$, we first show that $L(\beta)$ is a strictly monotone function of the budget to be allocated. That is, the value $L^*(B)$ $:= \min_\beta \{ L(\beta): \beta_p \ge 0 ~ \forall p \mbox{ and } \sum_p \beta_p \le B\}$ is a \emph{strict} monotone function of the budget $B$. Note that since on every edge $e$, $c_e(\beta_e)$ is \emph{strictly} monotone. Hence for every path $p$, $c_p(\beta_p)$ is strictly monotone as well.

\begin{clm}
If $L^*(B) > b_m$, then $L^*(B)$ is \emph{strictly} decreasing in $B$.
\label{clm:ppathsdecr}
\end{clm}

We first show the following claim.

\begin{clm}
Let $\beta'$, $\beta''$ be two vectors of allocations to paths in $\mP$ so that $\beta_p'' \ge \beta_p'$ for all $p \in \mP$ and the inequality strict for some $p$. Then if $L(\beta') > b_m$, then $L(\beta') > L(\beta'')$.
\label{clm:midecreasing}
\end{clm}

\begin{proof} The proof follows from the observation that $c_p(\beta'') \ge c_p(\beta')$ for every $p \in P_i$, with the inequality strict for at least one path. If $L(\beta'') \le b_m$, the claim is true since $L(\beta') > b_m$ by assumption. Otherwise,

\begin{align*}
L(\beta'') & = \frac{d + \sum_{p \in \mP} c_p(\beta_p'') b_p}{\sum_{p \in \mP} c_p(\beta_p'')} \\
	& = \frac{d + \sum_{p \in \mP} c_p(\beta_p') b_p + \sum_{p \in \mP} \left( c_p(\beta_p'') - c_p(\beta_p') \right) b_p }{\sum_{p \in \mP} c_p(\beta_p') + \sum_{p \in P_i} \left( c_p(\beta_p'') - c_p(\beta_p') \right)} \\
	& < L(\beta') 
\end{align*}

\noindent where the inequality follows from Fact~\ref{fact:xyz} and since $L(\beta') > b_i$ for all $i \in [m]$.
\end{proof}

\noindent \emph{Proof of Claim~\ref{clm:ppathsdecr}.} Let $B'$, $B'' \in \mathbb{R}_{> 0}$ with $B'' > B'$, and let $\beta'$ be the allocation that minimizes $L(\beta)$ subject to the total allocation being at most $B'$. Then consider the allocation $\beta''$ where $\beta_1'' = \beta_1' + (B''-B')$, and $\beta_i'' = \beta_i'$ on the other paths. By Claim~\ref{clm:midecreasing}, $L(\beta'') < L(\beta')$. Since $\beta''$ is a valid allocation for budget $B''$, $L^*(B'') \le L(\beta'') = L^*(B')$, and the claim follows. \qed

We now describe our algorithm. The algorithm proceeds by conducting a binary search for the optimal value $L^*(B)$. Initially, $b_m$ and $L(0)$ are our lower and upper bounds, and $\bar{L} = (b_m + L(0))/2$.

\begin{enumerate}
\item Let $\beta = 0$ be the initial allocation.
\item Increase the allocation to paths in $\mP$ so that for any path $p$, if $\beta_p > 0$, then 
\begin{equation}
(\bar{L}-b_p) \frac{d c_p(\beta_p)}{d \beta_p} \ge (\bar{L}-b_q) \frac{d c_q(\beta_q)}{d \beta_q} 
\label{eqn:barL}
\end{equation}
\noindent for all paths $q \in \mP$. Continue allocating in this manner until $L(\beta) = \bar{L}$. Note that by Claim~\ref{clm:midecreasing}, $L(\beta)$ is strictly decreasing in $\beta$, hence any process that monotonically increases allocations to paths will obtain $\bar{L}$ as long as $\bar{L} > b_m$.
\item Let $B' = \mone^T \beta'$, where $\beta'$ is the allocation obtained in Step 2. If $B' = B$, then $\beta'$ is the optimal allocation for budget $B$ and $L^* = \bar{L}$. If $B' > B$, then $L^*(B) > \bar{L}$, and $L^*(B) < \bar{L}$ otherwise.
\end{enumerate}

Step 2 in the algorithm can be implemented by binary search; we give details on the implementation in the Appendix. To show that the algorithm works, we now prove the correctness of Step 3. We start with the following claim about the allocation $\beta'$ obtained when Step 2 completes.

\begin{clm}
Let $B' = \mone^T \beta'$. Then the allocation $\beta'$ minimizes $L(\beta)$ for budget $B'$, i.e., $L(\beta') = L^*(B')$.
\label{clm:betaprime}
\end{clm}

\begin{proof}
The solution obtained satisfies

\[
(\bar{L}-b_p) \frac{d c_p(\beta_p')}{d \beta_p} \ge (\bar{L}-b_q) \frac{d c_q(\beta_q')}{d \beta_q}
\]

\noindent for all $p,q \in \mP$ with $\beta_p' > 0$. Since $\bar{L} = L(\beta')$, and $\sum_{p \in \mP} c_p(\beta') > 0$, this condition is equivalent to

\[
\frac{1}{\sum_{p \in P_i} c_p(\beta')} \left( L(\beta')-b_p \right) \frac{d c_p(\beta_p')}{d \beta_p} \ge \frac{1}{\sum_{p \in P_i} c_p(\beta')}  \left(L(\beta')-b_q\right) \frac{d c_q(\beta_q')}{d \beta_q}
\]

\noindent for all $p,q \in \mP$ with $\beta_p' > 0$, which are exactly the first order conditions for minimizing $L(\beta)$. Then by Lemma~\ref{lem:ppsufficient}, $\beta'$ must minimize $L(\beta)$ for budget $B'$.
\end{proof}

It follows from the claim that if $B' = B$, then $\beta'$ is the optimal allocation for $L(\beta)$ for budget $B$. If $B' > B$, note that by Claim~\ref{clm:ppathsdecr}, $L^*(B)$ is strictly monotone in $B$, and hence $\bar{L} = L^*(B') < L^*(B)$, and similarly if $B' < B$, then $\bar{L} > L^*(B)$. This proves the correctness of Step 3.

\subsection{A simple algorithm for parallel links}

We now consider the case where $G$ is a dipole graph, i.e., parallel edges between $s$ and $t$.  In contrast to the algorithm for the more general parallel paths case in Section~\ref{sec:ppaths}, we show that a very simple algorithm gives the optimal allocation in this case. We prove that there always exists an optimal solution where the entire budget is spent on a single edge. The algorithm for obtaining the optimal allocation is then straightforward: consider each edge in turn, and compute the delay at equilibrium obtained by allocating the entire budget to that edge. The optimal allocation is to allocate the budget to the edge for which the delay obtained is minimum.

As before, we assume every edge has flow at equilibrium in every valid allocation. From~(\ref{eqn:peflow}),

\begin{align}
L(\beta) & = \frac{d + \sum_{e \in E} c_e(\beta) b_e}{\sum_{e \in E} c_p(\beta)} \, . \label{eqn:pedelay2}
\end{align}

\noindent From the first-order conditions of optimality for~(\ref{eqn:main}), it follows that for an optimal allocation,~(\ref{eqn:foc}) must hold. We show that if two edges $e, e'$ have positive allocation in an optimal allocation $\beta$, then decreasing the allocation on one edge and increasing it on the other does not affect the delay at equilibrium. For $\delta \in \mathbb{R}$, let $\beta'$ be the allocation obtained by increasing the allocation to $e$ by $\delta$ and decreasing the allocation to $e'$ by $\delta$. 

\begin{lemma}
If $\frac{\partial L(\beta)}{\partial \beta_e} = \frac{\partial L(\beta)}{\partial \beta_{e'}}$, then $L(\beta) = L(\beta')$.
\label{lem:dipolechange}
\end{lemma}

We will use the expression for the delay at equilibrium in the following proof, obtained from~(\ref{eqn:pedelay2}) as

\begin{align}
\frac{\partial L(\beta)}{\partial \beta_e} & = \frac{\partial L(\beta)}{\partial c_e(\beta)} \frac{\partial c_e(\beta)}{\partial \beta_e} = -\frac{\mu_e}{\sum_{e' \in E} c_{e'}(\beta)} \left(L(\beta) - b_e \right) & \forall e \in E \label{eqn:perate}
\end{align}

\begin{proof}
Since $\frac{\partial L(\beta)}{\partial \beta_e} = \frac{\partial L(\beta)}{\partial \beta_{e'}}$, from~(\ref{eqn:perate}),

\begin{align*}
\mu_e \, \frac{b_e - L(\beta)}{\sum_{r \in E} c_r(\beta)} & = \mu_{e'} \, \frac{b_{e'} - L(\beta)}{\sum_{r \in E} c_r(\beta)} \, ,
\end{align*}

\noindent or, with some algebraic manipulation,

\begin{align}
L(\beta)  & = \frac{b_e \mu_e - b_{e'} \mu_{e'}}{\mu_e - \mu_{e'}} \label{eqn:peequal} \, .
\end{align}

Since in $\beta'$, the allocation to $e$ is increased and the allocation to $e'$ is decreased by $\delta$,

\begin{align}
L(\beta') & = \frac{d + \sum_{r \in E} b_r c_r(\beta) + b_e \mu_e \delta - b_{e'} \mu_{e'} \delta}{\sum_{r \in E} c_r(\beta) + \mu_e \delta - \mu_{e'} \delta} \, .\label{eqn:pebeta2}
\end{align}

\noindent From~(\ref{eqn:pedelay}) and~(\ref{eqn:peequal}), the numerator above is exactly $L(\beta)$ times the denominator. Replacing in~(\ref{eqn:pebeta2}) thus yields that $L(\beta') = L(\beta)$. 
\end{proof}

The following corollary is obtained since we can start with any optimal allocation that allocates to more than a single edge and by Lemma~\ref{lem:dipolechange} successively shift allocation onto a single edge so that we are left with an allocation on a single edge that yields the optimal delay at equilibrium.

\begin{corollary}
There exists an optimal allocation where the entire budget is allocated to a single edge.
\label{cor:paralleledges}
\end{corollary}

\begin{proof}
Consider an optimal allocation $\beta^*$ of the budget to $k > 1$ edges and edges $e$, $e'$ with strictly positive allocation. Consider the modified allocation $\beta'$: $\beta'_{r} = \beta^*_r$ for $r \neq e, e'$; $\beta'_e = \beta^*_e + \beta^*_{e'}$, and $\beta'_{e'} = 0$. Then $\beta'$ is a valid allocation, and since $\beta_{e}$, $\beta_{e'} > 0$, $\frac{\partial L(\beta)}{\partial \beta_e} = \frac{\partial L(\beta)}{\partial \beta_{e'}}$ by the first-order conditions for optimality. Then by the lemma, $L(\beta') = L(\beta^*)$. Thus, $\beta'$ is an optimal allocation where exactly $k-1$ edges have strictly positive allocation, and successively removing edges from the optimal allocation in this manner gives us the corollary.
\end{proof}

Thus, the simple algorithm given earlier that allocates the entire budget to a single edge is optimal.

\section{NP-Hardness in Series-Dipole Graphs}
\label{sec:sdipole}

In contrast to the previous section, we show that even in fairly simple networks called \emph{series-dipole networks}, the network improvement problem is NP-hard. A series-dipole graph consists of a number of subgraphs consisting of parallel edges (called \emph{dipole graphs}) connected in series. In fact, we show that even when each dipole consists of just two edges, computing the optimal allocation is NP-hard. We will use $n$ to denote the number of dipoles in the graph.

The delay at equilibrium in a series-dipole graph is the sum of delays on the individual dipoles. Further, given an allocation of the budget to dipoles rather than individual edges, by Corollary~\ref{cor:paralleledges} the optimal allocation to the edges can be determined by independently finding the edge in each dipole that minimizes the delay on the dipole on being allocated the entire budget for the dipole. Hence in this section we consider allocations to dipoles rather than individual edges, and define an allocation $\beta = (\beta_i)_{i \in [n]}$. Allocation $\beta$ is valid if $\sum_i \beta_i \le B$ and all $\beta_i \ge 0$. Further, define $L_i(\beta_i)$ as the optimal delay in dipole $i$ on being allocated $\beta_i$. Thus $L(\beta) = \sum_i L_i(\beta_i)$. 

We show that the problem of network improvement is NP-hard by a reduction from partition.

\begin{definition}[Partition]
Given $n$ items where item $i$ has value $v_i$ and $\sum_i v_i = 2V$, select a subset $S$ of the items so that $\sum_{i \in S} v_i = V$.
\end{definition}

For the $i$th dipole consisting of edges $e_1$ and $e_2$, the values for the parameters in our construction are as follows. Let $\delta = 19/31$ and $\lambda = 4 \sqrt{2}-1$. Then

\[
c_1 = \frac{\delta}{v_i}, ~ c_2 = \frac{1-\delta}{v_i}, ~ \mu_1 = \lambda v_i^2, ~ \mu_2 = 2 \lambda v_i^2, ~ b_1 = (\lambda + 2) v_i, ~ b_2 = 0, ~ \mbox{demand } d = 2(\lambda + 2)\, .
\]

\begin{clm}
For the instance constructed, there exists an allocation $\alpha_i = (1+\sqrt{2}) v_i$ for the $i$th dipole so that, for any allocation $x$, $L_i(x) \ge L_i(\alpha_i) + \alpha_i - x$, with equality if and only if $x = \alpha_i$ or $x = \alpha_i + v_i$.
\label{clm:dipole}
\end{clm}

\begin{proof}
Fix a dipole $i$. Let $v=v_i$ and $c_1$, $c_2$, $\mu_1$, $\mu_2$, $b_1$, $b_2$ and $d$ be the parameters given above. Since each dipole consists of two parallel edges, by Corollary~\ref{cor:paralleledges}, we only need to consider allocations to a single edge. Since $b_2 = 0$, if $L_i(x) \ge b_1$ then both edges carry flow at equilibrium, otherwise only edge $e_2$ has positive flow. Hence, from~(\ref{eqn:pedelay}),

\[
 L(x) = \left\{ \begin{array}{ll}
	\displaystyle \min \left\{ \frac{d+b_1\left(c_1 + (x/\mu_1)\right)}{c_1 + c_2 + x/\mu_1}, \frac{d+b_1 c_1}{c_1 + c_2 + (x/\mu_2)} \right\} & \mbox{if $L(x) \ge b_1$} \\
	\displaystyle  \frac{d}{c_2 + (x/\mu_2)} & \mbox{ otherwise}
\end{array} \right.
\]

\noindent Define the following functions:

\[
C_1(x) := \frac{d+b_1\left(c_1 + (x/\mu_1)\right)}{c_1 + c_2 + x/\mu_1}, \qquad C_2(x) := \frac{d+b_1 c_1}{c_1 + c_2 + (x/\mu_2)}, \mbox{ and } C_3(x) := \frac{d}{c_2 + (x/\mu_2)} \, .
\]

\noindent By definition, then $L(x) = \min \{C_1(x), C_2(x), C_3(x)\}$. Define $\gamma = 2v(5 \sqrt{2} +1)$ and $\alpha = v(\sqrt{2}+1)$. We will show the following properties for these functions:

\begin{enumerate}
\item $C_1(x) + x \ge \gamma$, with equality if and only if $x = \alpha$. \label{item:c1}
\item $C_2(x) + x \ge \gamma$, with equality if and only if $x = \alpha + v$. \label{item:c2}
\item $C_3(x) + x > \gamma$. \label{item:c3}
\end{enumerate}

\noindent The proof of the claim then follows from these properties, and from definition of $L(x)$.

\noindent \emph{Proof of (\ref{item:c1})}. The proof proceeds by showing that $C_1(x) +x$ is a strictly convex function for $x \ge 0$, and then showing that the function is minimized at $x = \alpha$ and $C_1(\alpha)+\alpha = \gamma$. For convenience, we differentiate $C_1(x) + x - b_1$, yielding

\[
\frac{d \left(C_1(x)+x - b_1\right)}{dx} = -\frac{1}{\mu_1} \frac{d - b_1 c_2}{\left( c_1 + c_2 + x/\mu_1 \right)^2} + 1
\]

\noindent It is obvious that the expression on the right is strictly increasing in $x$, and hence $C_1(x)$ is strictly convex. Setting the derivative to be zero gives

\[
x = \sqrt{\mu_1 ( d- b_1 c_2)} - \mu_1 (c_1 + c_2) \, .
\]

\noindent Replacing values for the parameters,

\begin{align*}
x & = \sqrt{v^2 \lambda \left( 2(\lambda+2) - v(\lambda+2)\frac{1-\delta}{v}\right)} - v^2 \lambda \frac{1}{v} \\
	& = v \left( \sqrt{2 \lambda (\lambda + 2) - \lambda (\lambda+2) \frac{12}{31}} - \lambda \right) \\
	& = v \left( \sqrt{2 \times 31 - 31 \frac{12}{31}} - 4 \sqrt{2} + 1 \right) ~ = ~   v (\sqrt{2} + 1) ~ = ~ \alpha
\end{align*}

\noindent where the second and third equalities are obtained by replacing the values for $t$ and $\lambda$ respectively. We now evaluate $C_1(\alpha)+\alpha$ to obtain

\begin{align*}
C_1(\alpha) + \alpha & = \frac{2(\lambda + 2) + v(\lambda+2) \left(\frac{\delta}{v} + \frac{v(1+\sqrt{2})}{v^2 \lambda}\right)}{\frac{1}{v}+\frac{v(1+\sqrt{2})}{v^2 \lambda}} + v(1+\sqrt{2}) \\
	& = v \frac{2(\lambda + 2) + (\lambda+2) \left(t + \frac{(1+\sqrt{2})}{\lambda}\right)}{1+\frac{1+\sqrt{2}}{\lambda}} + v(1+\sqrt{2}) \\
	& = v \frac{2\lambda (\lambda + 2) + \delta (\lambda+2) \lambda + (\lambda+2)(1+\sqrt{2})}{\lambda+1+\sqrt{2}} + v(1+\sqrt{2}) \\
	& = v \frac{62 + 19 + 9 + 5 \sqrt{2}}{5 \sqrt{2}} + v(1+\sqrt{2}) \\
	& = v (1 + 9 \sqrt{2}) + v (1 + \sqrt{2}) ~ = ~ v(2 + 10 \sqrt{2}) ~ = ~ \gamma \, .
\end{align*}

Thus, $C_1(x) + x$ is minimized at $x = \alpha$, and $C_1(\alpha) + \alpha = \gamma$. Since $C_1(x)$ is strictly convex for $x \ge 0$, $C_1(x) + x > \gamma$ for $x \neq \gamma$ and $x \ge 0$.

\noindent \emph{Proof of (\ref{item:c2})}. Our proof is very similar to the proof for $C_1(x)$. We observe that $C_2(x)+x$ is strictly convex for $x \ge 0$. We show that $C_2(x) + x$ is minimized at $x = \alpha + v$, and that $C_2(\alpha+v)+\alpha+v = \gamma$. By strict convexity, $C_2(x) + x > \gamma$ for $x \neq \alpha +v $ and $x \ge 0$, completing the proof.

Differentiate $C_2(x) + x$ gives us

\[
\frac{d \left(C_2(x)+x\right)}{dx} = -\frac{1}{\mu_2} \frac{d + b_1 c_1}{\left( c_1 + c_2 + x/\mu_2 \right)^2} + 1
\]

\noindent Again, the expression on the right is strictly increasing, and hence $C_2(x)$ is strictly convex. Setting the derivative to be zero gives

\begin{align*}
x & = \sqrt{\mu_2 ( d+ b_1 c_1)} - \mu_2 (c_1 + c_2)  \\
	& = \sqrt{ 2 v^2 \lambda \left( 2 (\lambda+2) + v(\lambda+2)\frac{\delta}{v} \right) } - 2 v^2 \lambda \frac{1}{v} \\
	& = v \left( \sqrt{ 2 \lambda(\lambda+2) (2 +\delta)} - 2 \lambda \right) \\
	& = v \left( 9 \sqrt{2} - 2( 4\sqrt{2}-1) \right) ~ = ~ v (\sqrt{2}+2) ~ = ~ \alpha + v \, .
\end{align*}

We evaluate $C_2(\alpha+v) + \alpha+v$ to obtain

\begin{align*}
C_2(\alpha+v)+\alpha+v & = \frac{2(\lambda + 2) + v(\lambda+2) \frac{\delta}{v}}{\frac{1}{v}+ \frac{\alpha+v}{2 v^2 \lambda}} + \alpha + v \\
	& = v\frac{4 \lambda(\lambda+2) + 2\delta \lambda (\lambda+2)}{2 \lambda + \frac{\alpha+v}{v}} + \alpha+v \\
	& = v \frac{124 + \frac{62 \times 19}{31}}{8 \sqrt{2} - 2 + 2 + \sqrt{2}} + v(2 + \sqrt{2}) ~ = ~ v 9 \sqrt{2} + v (2 + \sqrt{2}) ~ = ~ \gamma \, .
\end{align*}

This completes the proof for (\ref{item:c2}).

\noindent \emph{Proof of (\ref{item:c3})}. We show that the minimum value of $C_3(x) + x$ is strictly larger than $\gamma$. Differentiating $C_3(x) + x$,

\[
\frac{d \left(C_3(x)+x\right)}{dx} = -\frac{1}{\mu_2} \frac{d}{\left(c_2 + x/\mu_2 \right)^2} + 1
\]

\noindent and hence, $C_3(x) + x$ is minimized when

\begin{align*}
x & = \sqrt{d \mu_2} - \mu_2 c_2  \, .
\end{align*}

\noindent Let $x'$ denote this value that minimizes $C_3(x) + x$. Then

\begin{align*}
C_3(x) + x & \ge \frac{d}{c_2 + x'/\mu_2} + x' \\
	& = \sqrt{d \mu_2} + \sqrt{d \mu_2} - \mu_2 c_2 \\
	& = 2 \sqrt{ 2(\lambda + 2) 2 v^2 \lambda} - 2 v^2 \lambda \frac{1-t}{v} \\
	& = 4 v \sqrt{31} - 2v (4 \sqrt{2} - 1) \frac{12}{31} ~ > ~ \gamma \, .
\end{align*} 
\end{proof}

We choose our budget $B = V + \sum_i \alpha_i$. Claim~\ref{clm:dipole} is illustrated in Figure~\ref{fig:hardness}, which depicts the optimal delay at equilibrium $L_i(x)$ in dipole $i$ as a function of the allocation $x$ to the dipole. By Corollary~\ref{cor:paralleledges}, for a single dipole it is always optimal to assign the entire budget to a single edge. Further, from~(\ref{eqn:pedelay2}) and the first-order conditions for optimality it can be obtained that the budgets for which it is optimal to assign to a fixed edge form a continuous interval, and the delay at equilibrium as a function of the allocation is convex in these intervals. This is depicted by the two convex portions of the curve in Figure~\ref{fig:hardness}. Our proof of hardness is based on observing that our construction from Claim~\ref{clm:dipole} puts the entire curve above the line $y = L_i(\alpha_i) + \alpha_i - x$ except at points $x = \alpha_i$ and $x = \alpha_i + v_i$ where the curve is tangent to the line.

\begin{figure}[h]
\centering
\psfrag{Li}[c]{$L_i(x)$}
\psfrag{Li1}[r]{$L_i(\alpha_i)$}
\psfrag{Li2}[r]{$L_i(\alpha_i) - v_i$}
\psfrag{ai1}[c]{$\alpha_i$}
\psfrag{ai2}[c]{$\alpha_i + v_i$}
\psfrag{x}[c]{$x$}
\includegraphics[scale=0.3]{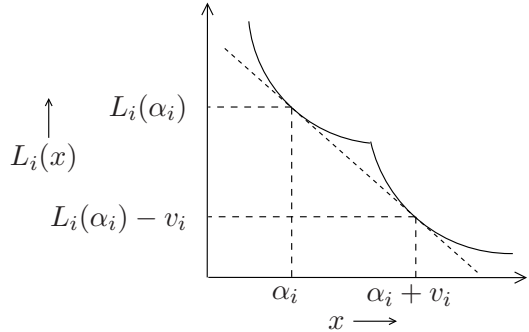}
\caption{The $L_i(x)$ curve for dipole $i$ is tangent to line $y = L_i(\alpha_i) + \alpha_i - x$ at exactly two points, $x = \alpha_i$ and $x = \alpha_i + v_i$.}
\label{fig:hardness}
\end{figure}

\begin{lemma}
The optimal delay for the constructed instance and the given budget is $\sum_i L_i(\alpha_i) - V$ if and only if the given instance of partition has a solution.
\label{lem:hardness}
\end{lemma}

\begin{proof}
For any allocation $\{\beta_i\}_{i \in [n]}$, the delay obtained at equilibrium is the sum of the dipoles. Thus $L(\beta) = \sum_i L_i(\beta_i)$. By Claim~\ref{clm:dipole}, $L(\beta) \ge \sum_i (L_i(\alpha_i) + \alpha_i - \beta_i)$. Since $\sum_i \beta_i \le B = V + \sum_i \alpha_i$, $L(\beta) \ge \sum_i L_i(\alpha_i) - V$. We will show that this lower bound is achieved if and only if the instance of partition has a solution.

Suppose the instance has a solution $S \subseteq [n]$. Then for the network improvement instance, allocate $\beta_i = \alpha_i$ to each dipole $i \not \in S$, and $\beta_i = \alpha_i + v_i$ to each dipole $i \in S$. For this allocation, by the claim,

\begin{align*}
L(\beta) & = \sum_{i \in S}  \left( L_i(\alpha_i) + \alpha_i - \alpha_i - v_i \right) + \sum_{i \not \in S}  \left( L_i(\alpha_i) + \alpha_i - \alpha_i \right) \\
	& = \sum_{i \in [n]} L_i(\alpha_i) - \sum_{i \in S} v_i ~ = \sum_{i \in [n]} L_i(\alpha_i) - V 
\end{align*}

\noindent completing the proof in this direction. For the other direction, suppose the instance of network improvement has an allocation $\beta$ that achieves the lower bound.  Then the inequality in Claim~\ref{clm:dipole} must hold with equality for each dipole. Hence for each dipole, either $\beta_i = \alpha_i$, or $\beta_i = \alpha_i + v_i$. Let $S$ be the set of dipoles with the latter allocation. Then 

\[
L(\beta) = \sum_{i \not \in S} L_i(\alpha_i) + \sum_{i \in S} (L_i (\alpha_i) - v_i) = \sum_i L_i(\alpha_i) - V
\]

\noindent where the first equality is by Claim~\ref{clm:dipole} and the second equality is because the allocation achieves the lower bound in the lemma. It follows immediately that $\sum_{i \in S} v_i = V$, and hence $S$ is a solution to the partition instance. \noindent
\end{proof}

\section{An FPTAS for Series-Parallel Graphs}
\label{sec:sepa}

In the previous section, we have shown that even for limited network topologies, obtaining the optimal allocation is weakly NP-hard, and thus these topologies are unlikely to have optimal polynomial time algorithms. We now show that for a large class of graphs with a single source and sink, we can obtain in polynomial time near-optimal algorithms for network improvement. Specifically, we present an FPTAS\footnote{A fully polynomial-time approximation scheme (FPTAS) is a sequence of algorithms $\{A_\epsilon\}$ so that, for any $\epsilon > 0$, $A_\epsilon$ runs in time polynomial in the input and $1/\epsilon$ and outputs a solution that is at most a $(1+\epsilon)$ factor worse than the optimal solution.} for the network improvement with bounded polynomial delays on a two-terminal series-parallel graph. Our algorithm is based on appropriately discretizing the space of budgets and flows on each subgraph $H$ of the graph $G$. We then obtain the allocation and flow that minimizes the maximum delay over all paths with positive flow in this discretized space. Although the discretized flow obtained will not be an equilibrium flow, by Lemma~\ref{lem:sepaeqmin}, this maximum delay will be an upper bound on the delay of the equilibrium flow for the allocation obtained. We will show that the delay of the discretized flow is a good approximation to the optimal delay.

We start with a definition for series-parallel graphs and of the corresponding decomposition tree. 

\begin{definition}[Series-parallel graph]
A single edge $e=(s,t)$ is a series-parallel graph with source $s$ and sink $t$. Further, if two graphs $G_1$ and $G_2$ are series-parallel graphs with source and sink $s_1$, $t_1$ and $s_2$, $t_2$ respectively, then they can be combined to form a new series-parallel graph by either of the following operations:
\begin{enumerate}
\item Series Composition: Merge $t_1$ and $s_2$ to obtain graph $G$, and let $s_1$ and $t_2$ be the new source and sink of $G$;
\item Parallel Composition: Merge $s_1$ and $s_2$ to obtain a new source $s$, and $t_1$ and $t_2$ to obtain a new sink $t$ in the resulting graph $G$.
\end{enumerate}
\end{definition}

The recursive definition of a series-parallel graph naturally yields a binary tree called the \emph{decomposition tree} of the series-parallel graph. The root of the tree corresponds to graph $G$, and the leaves correspond to the edges of $G$. Each internal node corresponds to the subgraph obtained by the series or parallel composition of the subgraphs corresponding to the children. In the following discussion, a subgraph of series-parallel graph $G$ is a graph obtained during the recursive construction, and hence corresponds to a node in the decomposition tree. We use $s_H$ and $t_H$ to denote the source and sink of subgraph $H$. We use $H_1 \sim H_2$ to denote the series composition of $H_1$ and $H_2$, and $H_1 || H_2$ to denote the parallel composition.

We first show that in a series-parallel network, the equilibrium flow minimizes the maximum delay over all paths that have positive flow.

\begin{lemma}
Let $f$ be the equilibrium flow in routing game $\Gamma$ on series-parallel graph $G$, and $g$ be any $s$-$t$ flow of value $d$. Then $\max_{p: f_p >0} l_p(f) \le \max_{p:g_p > 0} l_p(g)$.
\label{lem:sepaeqmin}
\end{lemma}

\noindent We use the following result about flows in series-parallel graphs. We omit a formal proof, which follows immediately by induction on the graph.

\begin{lemma}
Let $G$ be a directed two-terminal series-parallel graph with terminals $s$ and $t$, and $f$, $g$ be two $s$-$t$ flows satisfying $|g| \ge |f|$ and $|g| > 0$. Then there exists an $s$-$t$ path $p$ such that $\forall e \in p$, $g_e > 0$ and $g_e \ge f_e$.
\label{lem:sepapaths}
\end{lemma}

\noindent \emph{Proof of Lemma~\ref{lem:sepaeqmin}.}
By Lemma~\ref{lem:sepapaths}, it follows that

\[
\max_{p: g_p > 0} l_p(g) \ge \min_p l_p(f) \, .
\]

\noindent Since $f$ is an equilibrium flow, all paths carrying positive flow have minimum delay, and the lemma follows. \qed

Given a parameter $\epsilon > 0$, define $\nu := \max_e n_e$ as the maximum exponent of the delay function on any edge. Let $\lambda := \epsilon^2/m^2$ be our unit of discretization, where as before $m = |E|$. For clarity of presentation, we assume that $1/\lambda$ is integral. For any subgraph $H$ of $G$ and $k \in \mathbb{Z}_+$, we define the set of discretized allocations $A_\epsilon(H,k)$ as the set of all valid allocations of budget $k B \lambda$ to edges in $H$, so that the allocation to each edge is either 0 or an integral multiple of $B \lambda$. We define $F_\epsilon(H,k)$ as the set of all valid $s_H$-$t_H$ flows on the edges of $H$ of value $k d \lambda $, that are additionally either zero or an integral multiple of $d \lambda$ on every edge. More formally,

\begin{align*}
A_\epsilon(H,k) & := \left\{ (\beta_e)_{e \in E(H)}: \forall e \in E(H), \, \beta_e = k_e B \lambda \mbox{ for } k_e \in \mathbb{Z}_+ \mbox{ and } \sum_{e \in E(H)} \beta_e = k B \lambda \right\} \\
F_\epsilon(H,k) & := \left\{ (f_e)_{e \in E(H)}: f \mbox{ is an $s_H$-$t_H$ flow of value $k d \lambda$ and } \forall e \in E(H), \, f_e = k_e d \lambda \mbox{ for } k_e \in \mathbb{Z}_+ \right\}
\end{align*}

As usual, let $\beta^*$ be the optimal allocation, and $f^*$ and $L^*$ be the equilibrium flow and delay for allocation $\beta^*$. We initially make the assumption that $\beta_e^* \ge \lambda B / \epsilon$ on every edge, and will remove this assumption later. We first show that optimizing over flows and allocations in the discretized space is sufficient to obtain a good approximation to the optimal delay.

\begin{lemma}
There exists flow $\hat{f} \in F_\epsilon(G, 1/\lambda)$ and allocation $\hat{\beta} \in A_\epsilon(G, 1/\lambda)$ that satisfy

\[
\max_{p: \hat{f}_p > 0} \sum_{e \in p} l_e (\hat{f}_e, \hat{\beta}_e) \le (1+\epsilon)^\nu L^* / (1-\epsilon)^\nu \, .
\]
\label{lem:epsclose}
\end{lemma}

We first show the following claim. The proof uses the fact that for any flow $f$ on an acyclic graph, a path-decomposition $\{f_p\}_{p \in P}$ of flow $f$ can be obtained so that at most $m$ paths $p \in P$ have strictly positive flow.

\begin{clm}
There exists flow $\hat{f} \in F_\epsilon(G, 1/\lambda)$ that satisfies $\hat{f}_e \le (1+\epsilon) f_e^*$ for all $e \in E$, and allocation $\hat{\beta} \in A_\epsilon(G, 1/\lambda)$ that satisfies $\hat{\beta}_e \ge \beta_e^* (1-\epsilon)$.
\label{clm:epsclose}
\end{clm}

\begin{proof}
We will assume we are given $f^*$ and $\beta^*$ and will construct $\hat{f}$ and $\hat{\beta}$ that satisfy the conditions of the claim. We start with the allocation $\hat{\beta}$. On each edge, we round $\beta_e^*$ down to the nearest multiple of $\lambda B$ to obtain $\hat{\beta}_e$ on these edges. On an abitrary edge $e_1$, we allocate $\hat{\beta}_{e_1} = B - \sum_{e \neq e_1} \hat{\beta}_e$. Note that by assumption, $1/\lambda$ is integral. Since the allocation to every other edge is an integral multiple of $B \lambda$, so is the allocation to edge $e_1$. Allocation $\hat{\beta}$ is obviously a valid allocation of budget $B$, and thus $\hat{\beta} \in A_\epsilon(G,1/\lambda)$. Further, since $\beta_e^* \ge B \lambda/\epsilon$ on every edge by assumption, $\epsilon \beta_e^* \ge \lambda B$, and thus on every edge

\begin{align*}
\hat{\beta}_e & \ge \beta_e^* - B \lambda  \ge \beta_e^* - \epsilon \beta_e^* = \beta_e^* (1- \epsilon) \, .
\end{align*}

\noindent Allocation $\hat{\beta}$ is thus the required allocation. 

To obtain flow $\hat{f}$, we start with a flow decomposition $\{f_p^*\}_{p \in P}$ so that at most $m$ paths in $P$ have $f_p^* > 0$. There is some path $p$ with $f_p^* \ge md \lambda/\epsilon$ in this decomposition. Call this path $q$. Then on every path except $q$, we round $f_p^*$ down to the nearest multiple of $d \lambda$ to obtain $\hat{f}_p$ for that path, and assign the remaining flow $d - \sum_{p \neq q} \hat{f}_p$ to path $q$. Since we assume $1/\lambda$ is integral and the flow on every other path is an integral multiple of $d \lambda$, so is the flow on path $q$. Flow $\hat{f}$ is then a flow of value $v$ with the flow on every edge an integer multiple of $d \lambda$, and hence $\hat{f} \in F_\epsilon(G,1/\lambda)$. Further, on every edge $e \not \in q$, $\hat{f}_e \le f_e^*$. Since $f_q^* \ge \epsilon md\lambda/\epsilon$,

\[
\hat{f}_q \le f_q^* + m d \lambda \le f_q^* + \epsilon f_q^* = (1+\epsilon) f_q^* \, .
\]

\noindent Hence for any edge $e \in q$,
	
\begin{align*}
\hat{f}_e & = \sum_{p \in P, p \neq q} \hat{f}_p + \hat{f}_q  \le \sum_{p \in P, p \neq q} f_p^* + (1+\epsilon) f_q^*  \le (1+\epsilon) f_e^* \, .
\end{align*}

\noindent Flow $\hat{f}$ is thus the required flow.
\end{proof}

\noindent \emph{Proof of Lemma~\ref{lem:epsclose}.} We will show the lemma is true for flow $\hat{f}$ and allocation $\hat{\beta}$ obtained in Claim~\ref{clm:epsclose}. Note that for any edge $e$, $\hat{f}_e > 0$ only if $f_e^* > 0$. Then for any path $p$ with $\hat{f}_e > 0$ on every edge $e \in p$,

\begin{align*}
\sum_{e \in p} l_e(\hat{f}_e, \hat{\beta}_e) & = \sum_{e \in p} \left(\frac{\hat{f}_e}{c_e + \mu_e \hat{\beta}_e}\right)^{n_e} + b_e \\
	& \le \sum_{e \in p} \left(\frac{ (1+\epsilon) f_e^*}{c_e + (1 - \epsilon) \mu_e \beta_e^*}\right)^{n_e} + b_e \\
	& \le \left(\frac{1+\epsilon}{1-\epsilon}\right)^\nu \left( \sum_{e \in p} \frac{f_e^*}{c_e + \mu_e \beta_e^*} + b_e \right) \\
	& = \left(\frac{1+\epsilon}{1-\epsilon}\right)^\nu L^* \, .
\end{align*}

\noindent where the last equality is because $f_e^* > 0$ for every edge $e \in p$. \qed

We utilise the recursive structure of series-parallel graphs to give a dynamic programming algorithm to obtain an optimal flow and allocation in the discretized spaces $F_\epsilon$ and $A_\epsilon$ respectively. For subgraph $H$ and $k,l \in [1/\lambda]$, define $D(H,k,l)$ recursively as follows.

\[
D(H,k,l) := \left \{ \begin{array}{ll}
		0, & \mbox{ if $l = 0$} \\
		l_e(dl\lambda,Bk\lambda), & \mbox{ if $l > 0$ and $H = e$} \\
		\min_{u \in [k] \cup \{0\}} \{ D(H_1,u,l) + D(H_2, k-u,l)\}, & \mbox{ if $H=H_1 \sim H_2$.} \\
		\min_{u \in [k] \cup \{0\}, \, v \in [l] \cup \{0\}} \{ \max \{ D(H_1,u,v), D(H_2,k-u, l-v) \} \}, & \mbox{ if $H = H_1 \parallel H_2$. } 
	\end{array} \right.
\]

\noindent From the definition, $D(H,k,l)$ corresponds to a division of flow $dl\lambda$ and budget $Bl \lambda$ between its subgraphs $H_1$ and $H_2$ so that the assignment to each subgraph is a multiple of $d \lambda$ for the flow and $B \lambda$ for the budget. Hence, by recursion, $D(H,k,l)$ corresponds to an allocation of flow and budget to each edge of $H$, so that the flow on each edge is an integral multiple of $d \lambda$ and the budget allocated to each edge is a multiple of $B \lambda$. Further the total budget allocated is $kB \lambda$ and the value of the flow in $H$ is $ld \lambda$. By dynamic programming on the series-parallel decomposition tree for $G$, $D(G,1/\lambda, 1/\lambda)$ can then be computed in time $O(m/\lambda^2) = O(m^3 /\epsilon^2)$. For the following lemma, for subgraph $H$, flow $f$ and allocation $\beta$, define

\[
l_H(f, \beta) = \max_{s_H-t_H \mbox{ paths } p: f_p > 0} \sum_{e \in p} l_e(f_e, \beta_e) \, .
\]

\begin{lemma}
For any subgraph $H$ and $k,l \in [1/\lambda]$,
\[
D(H,k,l) = \min_{f \in F_\epsilon(H,l), \, \beta \in A_\epsilon(H,k)} l_H(f, 
\beta) \, .
\]
\label{lem:dpopt}
\end{lemma}

We omit a formal proof of the lemma, which follows from the recursive definition of $D(H,k,l)$. We now show that if $\beta_e^* \ge \lambda B/\epsilon$ on every edge, the allocation obtained for $D(G, 1/\lambda, 1/\lambda)$ is a near-optimal allocation.

\begin{theorem}
Given $\epsilon' > 0$ and an instance of CNDP in a series-parallel graph so that the optimal allocation satisfies $\beta_e^* \ge \epsilon' B/(6 \nu m^2)$, we can obtain in time $O(m^3 \nu^2/{\epsilon'}^2)$ a valid allocation $\hat{\beta}$ so that the delay at equilibrium for this allocation is at most $(1+ \epsilon') L^*$.
\label{thm:sepa}
\end{theorem}

\begin{proof}
Choose $\epsilon = \epsilon'/(6 \nu)$, and let $\hat{f} \in F_\epsilon(G, 1/\lambda)$ and $\hat{\beta} \in A_\epsilon(G, 1/\lambda)$ be the flow and allocation that obtains delay $D(G, 1/\lambda, 1/\lambda)$. By the recursive construction given, $\hat{f}$ and $\hat{\beta}$ can be obtained in time $O(m^3 \nu^2/{\epsilon'}^2)$. By Lemmas~\ref{lem:dpopt} and~\ref{lem:epsclose}, 

\[ 
D(G, 1/\lambda, 1/\lambda) \le \left(\frac{1+\epsilon}{1-\epsilon}\right)^\nu L^* \le (1+3 \epsilon)^\nu L^* = (1+\frac{\epsilon'}{2 \nu})^\nu L^* \le (1+\epsilon') L^*\, .
\]

\noindent Let $f$ be the equilibrium flow with allocation $\hat{\beta}$, and $g = \hat{f}$. Then by Lemma~\ref{lem:sepaeqmin}, 

\[
L(f,\hat{\beta}) \le \max_{p : g_p > 0} l_p(g) = D(g, 1/\lambda, 1/\lambda) \le (1+\epsilon') L^* \, .
\]
\qed
\end{proof}

Given $\epsilon' > 0$, define $\alpha := \epsilon'/(6 \nu m^2)$. To remove the assumption that $\beta_e^* \ge \alpha B$ on every edge, consider an allocation $\tilde{\beta}_e := \beta_e^*(1 - \alpha) + \alpha B$. Then $\tilde{\beta}$ is a valid allocation. Let $\tilde{f}$ be the equilibrium flow for allocation $\tilde{\beta}$. Then $\tilde{\beta}$ satisfies the conditions for Theorem~\ref{thm:sepa}, and hence $L(f, \hat{\beta} \le (1 + \epsilon') L(\tilde{f}, \tilde{\beta})$. Further,

\[
L(\tilde{f},\tilde{\beta}) \le L(f^*, \tilde{\beta}) \le \frac{1}{(1-\alpha)^\nu} L(f^*, \beta^*) \le (1+\epsilon') L^* \,
\]

\noindent where the first inequality follows from Lemma~\ref{lem:sepaeqmin}, the second inequality follows from the proof of Lemma~\ref{lem:epsclose}, and the third by definition of $\alpha$ and $\nu$. Thus, $L(f, \hat{\beta} \le (1 + \epsilon')^2 L^*$ where $\beta^*$ is no longer restricted. Choosing $\epsilon'$ appropriately yields the required approximation ratio.

\bibliographystyle{plain}
\bibliography{ni-bib}

\appendix

\section*{Appendix}

\subsection*{Proofs from Section~\ref{sec:ppaths}}

\paragraph*{Removing assumption about all paths being used at equilibrium.} Recall we order paths so that $b_1 \le b_2 \le \dots \le b_m$, where $b_i$ is the length of path $i$ and there are $m$ $s$-$t$ paths, and assume the inequalities are strict. Define $\mP_i := \cup_{j \le i} p_j$, i.e., $P_i$ is the set of paths with length at most $b_i$. Thus, $P_1 \subset P_2 \subset \dots \subset P_m$. It is easy to see that the equilibrium flow for any allocation has strictly positive flow on exactly the edges in $P_i$ for some $i$.

Let $\mP_k$ be the set of paths with strictly positive flow at equilibrium. By a similar derivation as for~(\ref{eqn:pedelay}), the delay at equilibrium is given by

\begin{align}
L(\beta) & = \frac{d + \sum_{p \in \mP_k} c_p(\beta) b_p}{\sum_{p \in \mP_k} c_p(\beta)} \, . \label{eqn:ppdelay}
\end{align}

For an allocation $\beta$, since we do not \textit{a priori} know the set of paths used in the optimal solution, we will consider all possible sets of paths. To formalize this, for an allocation $\beta$, define 

\begin{equation}
M_i(\beta) := \frac{d + \sum_{p \in \mP_i} c_p(\beta) b_p}{\sum_{p \in \mP_i} c_p(\beta)} \, \label{eqn:M}
\end{equation}

\noindent For allocation $\beta$ if $\mP_i$ is the set of paths used in the resulting equilibrium, then by~(\ref{eqn:ppdelay}), $L(\beta) = M_i(\beta)$. We show that knowing $M_i(\beta)$ for each $i \in [m]$ also allows us to obtain $L(\beta)$.

\begin{clm}
For allocation $\beta$ if $b_i < M_i(\beta) \le b_{i+1}$, then $L(\beta) = M_i(\beta)$.
\label{clm:LMequivalent}
\end{clm}

\begin{proof}
We will show that if $b_i < M_i(\beta) \le b_{i+1}$ then there is an equilibrium flow $f$ with delay $M_i(\beta)$. By the uniqueness of equilibrium flow, the claim must then be true. 

For path $p \in \mP_i$, let $f_p = \left(M_i(\beta)-b_p\right)c_p(\beta)$, and $f_p = 0$ for $p \not \in \mP_i$. Then 

\begin{align*}
\sum_p f_p & = M_i(\beta) \sum_{p \in \mP_i} c_p(\beta) - \sum_{p \in \mP_i} c_p(\beta) b_p = d 
\end{align*}

\noindent where the second equality follows from the definition of $M_i(\beta)$ in~(\ref{eqn:M}). Further, each $f_p \ge 0$, since $M_i(\beta) \ge b_p$ for $p \in \mP_i$. Thus $f$ is a valid flow. To see that $f$ is also an equilibrium flow, note that for any path with strictly positive flow, the delay is exactly $f_p/c_p(\beta) + b_p$ $=M_i(\beta)$, while any path with zero flow has delay at least $b_{i+1} \ge M_i(\beta)$.
\end{proof}

For any budget $B$, define $M_i^*(B) := \min_\beta \{ M_i(\beta): \beta \ge 0, ~ \mone^T \beta \le B\}$, i.e., $M_i^*(B)$ is the minimal value of $M_i(\beta)$ over all valid allocations. Then if $M_i^*(B) > b_i$, replacing $L(\beta)$ by $M_i(\beta)$, $b_m$ by $b_i$, and $\mP$ by $\mP_i$, our algorithm in Section~\ref{sec:ppaths} returns the allocation $beta$ that minimizes $M_i(\beta)$. 

In order to use this algorithm, we need to know $k$ so that the set of paths used by the equilibrium for the optimal allocation is exactly $\mP_k$, since in this case the minima of $L(\beta)$ and $M_k(\beta)$ coincide. 

Since we do not know $k$, we use the following iterative algorithm. Start with $i=1$, and obtain the allocation $\beta' = \min_\beta M_i(\beta)$. If $M_i(\beta') \le b_{i+1}$, we stop; $\beta'$ is then the allocation that minimizes $L(\beta)$, and $k=i$. If $M_i(\beta') > b_{i+1}$, we increase $i$ and repeat the process.

To show the correctness of this process, we use the following lemma.

\begin{lemma}
For $i < m$ and any budget $B$, if $M_i^*(B) > b_{i+1}$, then $M_{i+1}^*(B) > b_{i+1}$. \label{lem:pppathorder}
\end{lemma}

\begin{proof}
Let $\beta'$ and $\beta''$ be valid allocations that minimize $M_i(\beta)$ and $M_{i+1}(\beta)$ respectively. We will prove the contrapositive. Then by assumption $M_{i+1}(\beta'') \le b_{i+1}$. Further,

\begin{align*}
M_i(\beta') \le M_i(\beta'') & = \frac{d + \sum_{p \in P_i} c_p(\beta'') b_p}{\sum_{p \in P_i} c_p(\beta'')} \\
 & \le \frac{d + \sum_{p \in P_{i+1}} c_p(\beta'') b_p}{\sum_{p \in P_{i+1}} c_p(\beta'')} = M_{i+1}(\beta'')
\end{align*}

\noindent where the first inequality is since $\beta'$ minimizes $M_i(\beta)$ and the second inequality follows from the contrapositive of Fact~\ref{fact:xyz}, by setting $k = b_{i+1}$ and $z = c_{i+1}(\beta'')$.
\end{proof}

To see that the process given above obtains an optimal allocation $\beta^*$, let $\mP_k$ be the set of paths used by equilibrium for the optimal allocation. Suppose our process ends for $i < k$, then it must have found an allocation $\beta'$ so that $M_i(\beta') \le b_{i+1}$. Since the process did not terminate at $i-1$,$M_{i-1}^*(B) > b_i$ and hence by Lemma~\ref{lem:pppathorder} $M_i(\beta') > b_i$. Thus $b_i < M_i(\beta') \le b_{i+1}$, and hence by Claim~\ref{clm:LMequivalent}, $L(\beta') = M_i(\beta') \le b_{i+1} < L(\beta^*)$, where the last inequality is because $P_k$ is the set of paths used by equilibrium for the optimal allocation $\beta^*$ and $b_k \ge b_{i+1}$. This is obviously a contradiction, since $L(\beta^*)$ is minimal.

Further, the process must stop for $i=k$, since there exists an allocation $\beta^*$ with $L(\beta^*) \le b_{i+1}$. The process is hence correct, and will obtain the optimal allocation.

\medskip

\paragraph*{Implementing Step 2 of the algorithm.} Given $\bar{L}$, our problem is to obtain an allocation $\beta$ that satisfies~(\ref{eqn:barL}) and $L(\beta) = \bar{L}$. To obtain such an allocation, we will use a binary search procedure together with a concave relaxation of the original problem for which the first-order conditions exactly correspond to~(\ref{eqn:barL}).

Consider the following optimization problem $P(B)$, with variables $(\beta_p)_{p \in \mP}$ and parametrized by the budget $B$:

\[
\max \sum_{p \in \mP} (\bar{L} - b_p) c_p(\beta_p)
\]

\noindent subject to $(\beta_p)_{p \in \mP}$ being a valid allocation for budget $B$. The first-order conditions for this problem are exactly~(\ref{eqn:barL}). Further, for each path $c_p(\beta_p)$ is a concave function, and hence the objective is concave. Thus for a given budget $B$ the problem can be solved to obtain the optimal allocation. We now show that by increasing the budget $B$, we can obtain a monotone solution $\beta$ to $P(B)$.

\begin{clm}
Let $B'' > B'$, and $\beta'$ is an optimal solution to $P(B')$. Then there is an optimal solution $\beta''$ to $P(B'')$ that satisfies $\beta_p'' \ge \beta_p'$ on all paths $p$.
\label{clm:relaxmonotone}
\end{clm}

\begin{proof}
Let $\alpha_p''$ be an optimal solution to $P(B'')$. If $\alpha_p'' \ge \beta_p'$ on all paths $p$, we are done. Otherwise, let $q$ be a path such that $\alpha_q'' < \beta_q'$. Since $B'' > B'$, there is a path $r$ such that $\alpha_r'' > \beta_r'$. Then, by Corollary~\ref{cor:cpconcave}, 

\begin{align*}
(\bar{L} - b_q) \frac{\partial c_q(\alpha_q'')}{\partial \beta_q} & \ge (\bar{L} - b_q) \frac{\partial c_q(\beta_q')}{\partial \beta_q} \, ,\\
(\bar{L} - b_r) \frac{\partial c_r(\alpha_r'')}{\partial \beta_r} & \le (\bar{L} - b_r) \frac{\partial c_r(\beta_r')}{\partial \beta_r}  \, .\\
\end{align*}

\noindent Further, by the first-order conditions for optimality, since $\beta_q' > 0$ and $\alpha_r'' > 0$,

\begin{align*}
(\bar{L} - b_q) \frac{\partial c_q(\beta_q')}{\partial \beta_q}& \ge (\bar{L} - b_r) \frac{\partial c_r(\beta_r')}{\partial \beta_r}  \, ,\\
(\bar{L} - b_r) \frac{\partial c_r(\alpha_r'')}{\partial \beta_r} & \ge (\bar{L} - b_q) \frac{\partial c_q(\alpha_q'')}{\partial \beta_q} \, .\\
\end{align*}

\noindent It immediately follows that all the above inequalities must be equalities, and further, for any path $p$,

\begin{align}
\frac{\partial c_p(\beta_p')}{\partial \beta_p} & = \frac{\partial c_r(\alpha_p'')}{\partial \beta_p}  \, . \label{eqn:monotone}
\end{align}

Consider the allocation $\beta''$, constructed as follows. Let $\beta'' = \beta'$ initially. Then for $p = 1, 2, \dots, m$ where $m$ is the number of paths, if $\alpha_p'' > \beta_p'$,  $\beta_p'' = \beta_p'' + \min \{\alpha_p'' - \beta_p', B'' - \sum_{q \in [m]} \beta_q'' \}$. Then $\beta_p'' \ge \beta_p'$ for all paths $p$. Further, $\beta_p''$ lies between $\beta_p'$ and $\alpha_p''$, and hence by concavity of conductance and~(\ref{eqn:monotone}), for all paths $p$,

\begin{align*}
\frac{\partial c_p(\beta_p'')}{\partial \beta_p} & = \frac{\partial c_p(\alpha_p'')}{\partial \beta_p}  \, .\\
\end{align*}

\noindent Also, since $\sum_q \beta_q' = B'$ and $\sum_q \max\{\alpha_q'', \beta_q'\} \ge B''$, it follows that $\sum_q \beta_q'' = B''$. Thus allocation $\beta''$ satisfies the first-order conditions for optimality of $P(B'')$, and since $P(B'')$ is a concave maximization program, $\beta''$ is also an optimal solution that satisfies the conditions of the claim.
\end{proof}

We now use the following binary search procedure. Let $B^l = 0$ and $B^h = mB$ be the lower and limits on the budget, and $\bar{B} = (B^l + B^h)/2$. Solve $P(\bar{B})$, and let $\bar{\beta}$ be the optimal solution. If $L(\bar{B}) < \bar{L}$, set $B^l = \bar{B}$, $\bar{B} = (B^l + B^h)/2$, add constraints $\beta \ge \bar{\beta}$ to $P(\bar{B})$ and solve again. Similarly if $L(\bar{B}) > \bar{L}$, we set $B^h = \bar{B}$, $\bar{B} = (B^l + B^h)/2$, add constraints $\beta \le \bar{\beta}$ to $P(\bar{B})$ and solve again.

Let $\beta'$ and $\beta''$ be the solution to $P(B')$ and $P(B'')$ with $B' > B''$. The added bounds on $\beta$ ensure that $\beta' \ge \beta''$, and since $\sum_p \beta_p' = B'$ at optimality, for some path $p$ $\beta_p' > \beta_p''$. By Claim~\ref{clm:midecreasing}, $L(\beta') < L(\beta'')$. Hence the binary search procedure must obtain a budget $\bar{B}$ and an allocation $\bar{\beta}$ so that $L(\bar{\beta}) = \bar{L}$. Further, by Claim~\ref{clm:relaxmonotone}, $\bar{\beta}$ satisfies the KKT conditions for the original problem, and hence satisfies~(\ref{eqn:barL}) as well.

\end{document}